\documentclass[10pt]{amsart}
\usepackage{amsfonts,amssymb}
\RequirePackage{ifpdf}
\usepackage{amsmath}
\usepackage{color}
\usepackage{pstcol,pst-node}
\usepackage{graphics}
\usepackage{epic}
\usepackage{curves}

\def\appendix#1{
\addtocounter{section}{1} \setcounter{equation}{0}
\renewcommand{\thesection}{\Alph{section}}
\section*{Appendix \thesection\protect\indent\quad
#1}
}
\renewcommand{\theequation}{\thesection.\arabic{equation}}

\catcode`\@=11
\def\marginnote#1{}

\newcount\hour
\newcount\minute
\newtoks\amorpm
\hour=\time\divide\hour by60 \minute=\time{\multiply\hour by60
\global\advance\minute by-\hour}
\edef\standardtime{{\ifnum\hour<12 \global\amorpm={am}%
        \else\global\amorpm={pm}\advance\hour by-12 \fi
        \ifnum\hour=0 \hour=12 \fi
        \number\hour:\ifnum\minute<10 0\fi\number\minute\the\amorpm}}
\edef\militarytime{\number\hour:\ifnum\minute<100\fi\number\minute}

\newcommand{\tcr}{\textcolor{red}}

\newcommand{\tcg}{\textcolor{green}}

\let\wtd=\widetilde

\def\draftlabel#1{{\@bsphack\if@filesw {\let\thepage\relax
      \xdef\@gtempa{\write\@auxout{\string
          \newlabel{#1}{{\@currentlabel}{\thepage}}}}}\@gtempa \if@nobreak
    \ifvmode\nobreak\fi\fi\fi\@esphack} \gdef\@eqnlabel{#1}}
    \def\@eqnlabel{}
\def\@vacuum{}
\def\draftmarginnote#1{\marginpar{\raggedright\scriptsize\tt#1}}

\def\draft{
  \oddsidemargin -.5truein
  \def\@oddfoot{\footnotesize \sl preliminary draft \hfil
    \rm\thepage\hfil\sl\today\quad\militarytime}
  \let\@evenfoot\@oddfoot \overfullrule 3pt
    \let\label=\draftlabel
    \let\marginnote=\draftmarginnote
  \def\@eqnnum{(\theequation)\rlap{\kern\marginparsep\tt\@eqnlabel}%
    \global\let\@eqnlabel\@vacuum}

  }
\newcommand{\tr}{\,{\rm Tr}\,}

\def\be{\begin{equation}}
\def\ee{\end{equation}}
\def\bea{\begin{eqnarray}}
\def\eea{\end{eqnarray}}
\def\<{\langle}
\def\>{\rangle}

\def\nn{\nonumber}
\def\HH{{\mathbb H}}
\def\RR{{\mathbb R}}
\def\PP{{\mathbb P}}
\def\ZZ{{\mathbb Z}}
\def\F{{\mathcal F}}
\let\wtd=\widetilde

\def\one#1{#1^{\raise5pt\hbox{$\scriptstyle\!\!\!\!1$}}\,{}}
\def\two#1{#1^{\raise5pt\hbox{$\scriptstyle\!\!\!\!2$}}\,{}}
\def\onetwo#1{#1^{\raise5pt\hbox{$\scriptstyle\!\!\!\!\!{12}$}}\,{}}

\def\e{e}
\def\x{{\mathbf x}}
\def\pp{{\mathbf p}}
\def\A{{\mathcal A}}
\def\Q{{\mathbb Q}}

\newtheorem{theorem}{Theorem}[section]
\newtheorem{lm}[theorem]{Lemma}
\newtheorem{prop}[theorem]{Proposition}
\newtheorem{corollary}[theorem]{Corollary}
\theoremstyle{definition}
\newtheorem{df}[theorem]{Definition}

\newtheorem{remark}[theorem]{Remark}

\theoremstyle{remark}

\newtheorem{conjecture}[theorem]{Conjecture}

\allowdisplaybreaks

\begin{document}
\title[Teichm\"uller spaces of Riemann surfaces with orbifold points]
{Teichm\"uller spaces of Riemann surfaces with orbifold points of arbitrary order and
cluster variables}
\author{Leonid Chekhov$^{\ast,\dag}$}\thanks{$^{\ast}$Steklov Mathematical Institute and  Laboratoire Poncelet,
Moscow, Russia}\thanks{$^{\dag}$School of Mathematics, Loughborough University, UK.}
\author{Michael Shapiro$^\diamondsuit$}\thanks{$^\diamondsuit$Mathematical Department, Michigan State University, East Lansing, USA}

\maketitle

\begin{abstract}
We define  a new generalized class of cluster type mutations
for which exchange transformations are given by reciprocal polynomials.
In the case of second-order polynomials of the form $x+2\cos{\pi/n_o}+x^{-1}$ these transformations are related to
triangulations of Riemann surfaces of arbitrary
genus with at least one hole/puncture and with an arbitrary number of
orbifold points of arbitrary integer orders $n_o$. In the second part of the paper,
we propose the dual graph description of the corresponding Teichm\"uller spaces, construct the Poisson
algebra of the Teichm\"uller space coordinates, propose the
combinatorial description of the corresponding geodesic functions and
find the mapping class group transformations thus providing the complete description of the above
Teichm\"uller spaces.
\end{abstract}

\section{Introduction}

Since their appearance, cluster variables~\cite{FZ} find applications in geometry. An important example
of the cluster variables  is provided by
$\lambda$-lengths~\cite{ThSh},~\cite{Penn1} of curves that partitions Riemann surfaces with punctures into ideal triangles. In this case, exchange polynomials are quadratic. These coordinates were
generalized in~\cite{Fock1},~\cite{Fock2} to the case of Riemann surfaces with holes. At the same time,
a combinatorial description of geodesic functions in terms of the dual variables, the shear coordinates,
as well as their quantization, was developed in~\cite{ChF}. Amazingly enough, transition from punctures to holes
does not effectively change the corresponding cluster algebra. Generalizations of Teichm\"uller spaces of Riemann
surfaces to the case of bordered Riemann surfaces~\cite{KaufPen} or ciliated Riemann surfaces~\cite{FG} were constructed.
The corresponding cluster algebras were developed in~\cite{GSV1,FG1,FST}, whereas the geometrical pattern underlying the
bordered Riemann surfaces was identified with that of Riemann surfaces with ${\mathbb Z}_2$-orbifold points in~\cite{Ch1},~\cite{Ch1a},
where the corresponding mutations (flips) in terms of the shear coordinates were constructed. These flips preserve the sets
of geodesic functions; the corresponding transformations for cluster variables were considered in~\cite{FST-Orb} and
the corresponding mutations were
again given by the standard two-term relations. In \cite{ChM}, the description of Teichm\"uller spaces of
Riemann surfaces with holes and with orbifold points of order two and three was given.

%It was, however, an ambiguity in identifying a boundary component with an orbifold point; in the present paper, we
%demonstrate that this ambiguity is due to choice of the type of an orbifold point corresponding to the given boundary component; we therefore
In the present paper we provide the combinatorial description of Riemann surfaces with holes and with orbifold points of arbitrary orders. We show that mutations
for orbifold points of order greater than two are given by three-term transformations (unlike the  two-term transformations for order two) determined by a
second-order reciprocal polynomial.  We prove the Laurent phenomenon and positivity property for these transformations. The positive coefficients
of Laurent polynomials however
are not necessarily integral in the presence of orbifold points of order greater than three.
On the shear-coordinate
side, we define the complete set of real-valued coordinates,
construct all the geodesic functions for such surfaces, all the mapping-class-group transformations, and prove the regularity
condition, that is, that all elements of the corresponding Fuchsian group are hyperbolic or parabolic ones except elements
conjugate to loops around orbifold points. We therefore have a regular (up to exactly the indicated orbifold points) Riemann surface
with holes for any choice of the introduced real coordinates, and vice versa; these coordinates parametrize therefore the corresponding
Teichm\"uller spaces of Riemann surfaces with holes and with orbifold points of arbitrary orders.

As in the original formulation of cluster algebras, the insight into 
orbifold triangulations helps us to formulate a more general construction.
In \cite{GSV}, particular generalizations of cluster transformations were described that preserve Poisson bracket and have additionally some universal properties.
Until recently no applications of these transformations were known. In this paper we compute that mutations of orbifold triangulations are examples of generalization~\cite{GSV}.
Another example of  generalized quadratic cluster mutations appear under the name quasi-cluster algebra associated with non-orientable surfaces
in preprint~\cite{DP}.
Motivated by that we propose a new algebraic construction of
generalized cluster algebras with mutations given by reciprocal polynomials of arbitrary order.
Using the tools of the standard cluster algebra \cite{FZ1}, \cite{FZ2}, we prove that the Laurent phenomenon holds true in this case
as well. For algebras of order greater than two, we do not know whether the positivity property holds in general; it however holds in all
tested examples, so we formulate it as a conjecture. We also prove that
generalized cluster algebras of finite type satisfy the same Cartan--Killing classification as the standard cluster algebras. Suggested construction is a particular case of more general construction of \cite{LP}. However, we note that generally speaking mutations in \cite{LP} preserve neither presymplectic 2-form nor the Poisson bracket.

\vskip 2mm \noindent{\bf Acknowledgements.}  The authors are
grateful to Anna Felikson, Pavel Tumarkin, Sergey Fomin, and Dylan Thurston for many enlighting conversations and, specially,  to Alek Vainshtein for
valuable comments improving our paper.

The work of L.Ch. was supported in part 
by the Russian Foundation for Basic Research (Grant Nos. 11-01-00440-a and 11-01-12037-ofi-m-2011),
by the Grant of Supporting Leading
Scientific Schools of the Russian Federation NSh-4612.2012.1, and by the Program Mathematical Methods for Nonlinear Dynamics.

Michael Shapiro was supported in part by grants DMS-0800671 and DMS-1101369.

\section{Generalized cluster algebra}\label{s-algebra}

\noindent We briefly remind the definition of cluster algebra.

An integer $n\times n$ matrix $B$ is called \emph{skew-symmetrizable} if there exists an
integer diagonal $n\times n$ matrix $D=diag(d_1,\dots,d_n)$,
such that the product $BD$ is a skew-symmetric matrix, i.e., $b_{ij}d_j=-b_{ji}d_i$.

{
\color{black} Let $\PP$ be \emph{a semi-field } equipped with commutative multiplication $\cdot$ and addition $\oplus$. We assume that the multiplicative group of $\PP$ is a free abelian group. $\PP$ is \emph{a coefficient group} of cluster algebra. $\ZZ\PP$  is the integer group ring, $\F$ is a field of rational functions in $n$ independent
variables with coefficients in the field of fractions of $\ZZ\PP$.
$\F$ is called an ambient field.

\begin{df}
\emph{A seed} is a triple $(\x,\pp,B)$, where
\begin{itemize}
\item $\pp=(p_{x}^\pm)_{x\in\x}$, a $2n$-tuple of elements of $\PP$ is a \emph{coefficient tuple} of cluster $\x$;
\item $\x=\{x_1,\dots,x_n\}$ is a collection of algebraically independent rational functions of $n$ variables which generates $\F$ over the field of fractions of $\ZZ\PP$;
\item $B$ is a skew-symmetrizable \emph{exchange matrix}.
\end{itemize}
The part $\x$ of seed $(\x,\pp,B)$ is called \emph{cluster}, elements $x_i\in\x$ are called \emph{cluster variables},
and $B$ is called \emph{exchange matrix}.

\end{df}

\begin{df}[seed mutation]
For any $k$, $1\le k\le n$ we define \emph{the mutation} of seed $(\x,\pp,B)$ in direction $k$
as a new seed $(\x',\pp',B')$ in the following way:
\begin{equation}\label{eq:MatrixMutation}
b'_{ij}=\left\{
           \begin{array}{ll}
             -b_{ij}, & \hbox{ if } i=k \hbox{ or } j=k; \\
             b_{ij}+\frac{|b_{ik}|b_{kj}+b_{ik}|b_{kj}|}{2}, & \hbox{ otherwise.}
           \end{array}
         \right.
\end{equation}

\begin{equation}\label{eq:ClusterMutation}
x'_i=\left\{
           \begin{array}{ll}
             x_i, & \hbox{ if } i\ne k; \\
             \frac{p^+_k\prod_{b_{kj}>0}x_j^{b_{kj}}+
             p^-_k\prod_{b_{kj}<0} x_j^{-b_{kj}}}{x_k}, & \hbox{ otherwise.}
           \end{array}
         \right.
\end{equation}

\begin{eqnarray}\label{eq:CoeffMutation}
 \nonumber % to remove numbering (before each equation)
  p'^\pm_k &=& p^\mp_k \\
  \hbox{ for } i\ne k\qquad p'^+_i/p'^-_i &=& \left\{
           \begin{array}{ll}
             (p^+_k)^{b_{ik}}p^+_i/p^-_i, & \hbox{ if } b_{ik}\ge 0; \\
             (p^-_k)^{b_{ik}}p^+_i/p^-_i, & \hbox{ if } b_{ik}\le 0; \\
           \end{array}
         \right.
\end{eqnarray}

\end{df}

\noindent
We write $(\x',\pp',B')=\mu_k\left((\x,\pp,B)\right)$.
Notice that $\mu_k(\mu_k((\x,\pp,B)))=(\x,\pp,B)$.
We say that two seeds are \emph{mutation-equivalent}
if one is obtained from the other by a sequence of seed mutations.
Similarly we say that two clusters or two exchange matrices are \emph{mutation-equivalent}.

For any skew-symmetrizable matrix $B$ we define \emph{initial seed} $$(\x,\pp,\!B)=(\!\{x_1,\dots,x_n\}\!,\!\{p_1^\pm,\ldots,p_n^\pm\}\!,\!B),$$ where $B$ is the \emph{initial exchange matrix}, $\x=\{x_1,\dots,x_n\}$ is the \emph{initial cluster}, $\pp=\{p_1^\pm,\ldots,p_n^\pm\}$ is the \emph{initial coefficient tuple}.
}

{\it Cluster algebra} $\A(B)$ associated with the skew-sym\-met\-ri\-zab\-le $n\times n$ matrix $B$ is a subalgebra of $\Q(x_1,\dots,x_n)$ generated by all cluster variables of the clusters mutation-equivalent
to the initial seed $(x,B)$.

Cluster algebra $\A(B)$ is called \emph{of finite type} if it contains only finitely many
cluster variables. In other words, all clusters mutation-equivalent to initial cluster contain
totally only finitely many distinct cluster variables.

% Clearly that any finite type cluster algebra is of finite mutation type. However, there are examples of infinite type cluster algebras that are of finite mutation type.

%\begin{equation}\label{matrix-mutation}
%    b'_{kl}=\left\{
%              \begin{array}{ll}
%                -b_{kl}, & \hbox{ if } k=l \hbox{ or } j=l; \\
%                b_{kl}+\frac{1}{2}\left(|b_{kl}|b_{lj}+b_{kl}|b_{lj}|\right), & \hbox{ otherwise.}
%              \end{array}
%            \right.
%\end{equation}

Two most important properties of cluster algebra are Laurent phenomenon~\cite{FZ1} and finite type classification~\cite{FZ2}.
More exactly, Laurent phenomenon states that any cluster variable is expressed as a Laurent polynomial in terms of the initial cluster.
The remarkable finite type classification claims that  cluster algebras of finite type are in one-to-one correspondence
with the Dynkin diagrams of finite type.

\subsection{Generalized cluster transformations}

Now we introduce more general  cluster transformations.

%{\color{red} Nado vvesti coeffs, no esche ne jasno kak.}

Assume that $B$ is a skew-symmetrizable integer matrix such that all elements in its $k$th row are divisible by $d_k$.
Define $\beta_{kj}=b_{kj}/d_k$.

\begin{lm}
Let $B'=\mu_l(B)$ be obtained from $B$ by mutation in direction $l$. Then, all entries $b'_{kj}$ of $k$-th row of $B'$ are divisible by $d_k$.
\end{lm}

\begin{proof} The statement follows immediately from matrix mutation~\ref{eq:MatrixMutation}.
\end{proof}

We now fix $d_k$ for all $k$ from $1$ to $n$ and assume that all elements $b_{kj}$ of $k$th row of integer skew-symmetrizable matrix $B$ are divisible by $d_k$.

For a collection $\pp_i=(p_{i;0},\ldots,i_{i;d_i})$ we define the {\em exchange polynomial}
 $\theta_i[\pp_i](u,v)=\sum_{\ell=0}^{d_i} p_{i;\ell} u^\ell v^{d_i-\ell}$ be a polynomial of degree $d_i$. The corresponding inhomogeneous polynomial we denote by $\rho_i[\pp_i]=\rho_i[\pp_i](t)=\sum_{\ell=0}^{d_i} p_{i;\ell} t^\ell$. Note that,
$\theta_i[\pp_i](u,v)=u^{d_i}\rho_i[\pp_i](v/u)$.

We define a generalized seed $q$ of a generalized cluster algebra as a triple\newline  $q=(\x(q),\overline{\pp(q)},B(q))$, where
$\x=(x_1(q),\ldots,x_n(q))$ is a $n$-tuple of cluster variables in seed $q$, $\overline{\pp(q)}=(\pp_1(q),\ldots,\pp_n(q))$ is $n$-tuple of coefficient collections $\pp_i(q)$,
 $\pp_i(q)=(p_{i;0}(q),\ldots, p_{i;d_i}(q))$ is a $d_i+1$-tuple
of coefficients of $\theta_i[q]$, and, finally, $B(q)$ is an exchange $n\times n$ matrix.

Generalized cluster mutations are described by the following formulas:

Exchange matrix is mutated in direction $k$ by the rule ~\ref{eq:MatrixMutation}, which
therefore depends only on the degree of the exchange polynomial and not on its coefficients. We
introduce $u_{j;>0}=\prod_{\beta_{j,\ell}>0} x_\ell^{\beta_\ell}$, $u_{j;<0}=\prod_{\beta_{j,\ell}<0} x_\ell^{-\beta_\ell}$,
Mutation of cluster variables is given by the rule $(\{x_i'\},\{\pp_i'\},B')=\mu_k(\{x_i\},\{\pp_i\},B)$:

\begin{equation}\label{eq:GenClusterMutation}
x'_i=\left\{
           \begin{array}{ll}
             x_i, & \hbox{ if } i\ne k; \\
             \frac{\theta_k(u_{k;>0},u_{k;<0})}{x_k}, & \hbox{ otherwise.}
           \end{array}
         \right.
\end{equation}

Coefficients mutate by the following generalized rule:
\begin{eqnarray}\label{eq:GenCoeffMutation}
 \nonumber % to remove numbering (before each equation)
  p'_{k;\ell} &=& p_{k;d_k-\ell} \\
  \hbox{ for } i\ne k\qquad p'_{i;j}/p'_{i;0} &=& \left\{
           \begin{array}{ll}
             (p_{k;d_k})^{j\beta_{ik}}p_{i;j}/p_{i;0}, & \hbox{ if } b_{ik}\ge 0; \\
             (p_{k;0})^{j\beta_{ik}}p_{i;j}/p_{i;0}, & \hbox{ if } b_{ik}\le 0; \\
           \end{array}
         \right.
\end{eqnarray}

\begin{remark}
Note that if we assume that coefficients of all $\theta_i$ do not change under mutation then the corresponding inhomogeneous polynomial $\rho_i$ is reciprocal of degree $d_i$, i.e., $t^{d_i}\rho(1/t)=\rho(t)$.
\end{remark}

\begin{theorem}\label{thm:generalizedLaurent} (Laurent property) Any generalized cluster variable is a Laurent polynomial in initial cluster variables $x_i$.
\end{theorem}

\begin{proof}
The proof uses the "caterpillar lemma"~\cite{FZ1}.

\begin{lm}\label{lm:caterpillar}
Assume that a generalized exchange pattern on $T_{n,m}$ satisfies the following conditions:
\begin{enumerate}
\item For any edge
the polynomial $P$ does not depend on $x_k$ and is not divisible by any $x_i$, $i\in[n]$.
\item Each exchange polynomial has nonnegative coefficients
\item For any three edges
labeled by $i,j,i$\newline
\smallskip
{\psset{unit=0.7}
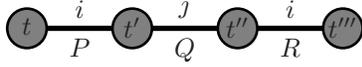
\begin{pspicture}(-3,0.8)(3,-0.8)
\pcline[linewidth=2pt](-3,0)(3,0)
\pscircle[linewidth=1pt,fillstyle=solid,fillcolor=gray](-3,0){0.4}
\pscircle[linewidth=1pt,fillstyle=solid,fillcolor=gray](-1,0){0.4}
\pscircle[linewidth=1pt,fillstyle=solid,fillcolor=gray](1,0){0.4}
\pscircle[linewidth=1pt,fillstyle=solid,fillcolor=gray](3,0){0.4}
\rput(-2,0.2){\makebox(0,0)[cb]{$i$}}
\rput(-2,-0.2){\makebox(0,0)[ct]{$P$}}
\rput(0,0.2){\makebox(0,0)[cb]{$j$}}
\rput(0,-0.2){\makebox(0,0)[ct]{$Q$}}
\rput(2,0.2){\makebox(0,0)[cb]{$i$}}
\rput(2,-0.2){\makebox(0,0)[ct]{$R$}}
\rput(-3,0.2){\makebox(0,0)[ct]{$t$}}
\rput(-1,0.2){\makebox(0,0)[ct]{$t'$}}
\rput(1,0.2){\makebox(0,0)[ct]{$t''$}}
\rput(3,0.2){\makebox(0,0)[ct]{$t'''$}}
\end{pspicture}
}
\smallskip\newline
we have
$L\cdot Q_0^b\cdot P=R|_{x_j\leftarrow\frac{Q_0}{x_j}},$
where $b$ is a negative integer, $Q_0=Q|_{x_i=0}$, and $L$ is a Laurent monomial whose coefficient lies in $\mathbb A$
and is coprime with $P$.
\end{enumerate}
Then each element $x_i(t)$ for $i\in[n]$, $t\in T_{n,m}$ is a Laurent polynomial in $x_1(t_0),\ldots,x_n(t_0)$ with coefficients in $\mathbb A$.
\end{lm}

By definition of generalized cluster transformation $P=\theta_i[t],\ Q=\theta_j[t'],\ R=\theta_i[t'']$.

Note that parts (1) and (2) are evidently satisfied by generalized cluster mutations.
It remains to proof part (3). If $x_i$ is not included into any monomial of $Q$ then generalized mutation with labels $i$ and $j$ commute and the latter mutation is inverse to the former, namely, $x_i(t''')=x_i(t)$.

We consider the case where $x_i$ enters a monomial of $Q$. For simplicity we denote $b_{rs}(t')$ ($\beta_{rs}(t')$) by $b'_{rs}$ ($\beta'_{rs}$, resp.) and, specifically, $b_{ji}(t')$ by $a$.

By our assumptions $a\ne 0$. Moreover, without loss of generality we can assume that $a>0$, otherwise we replace $B$ by $-B$. Since $Q$ is determined by the homogeneous polynomial of two variables where only one variable contains a positive power of $x_i$ then $x_i$ enters all monomial of $Q$ but one. Hence,
$Q_0=Q|_{x_i=0}$ is a monomial.
Moreover, $Q_0=p_{j;0}\prod_{b'_{jk}<0} x_k^{-b'_{jk}}$.
Note that $P=\theta_i(u_{>0}(t'), u_{<0}(t'))$.
By the mutation rule~\ref{eq:ClusterMutation},
\begin{equation*}
    b_{il}(t'')=\left\{
                  \begin{array}{ll}
                    b'_{il}, & \hbox{ if } b'_{jl}\ge 0;\\
                    b'_{il}(t')-b'_{ij} b'_{jl}, & \hbox{ otherwise.}
                  \end{array}
                \right.
\end{equation*}

By the definition of generalized cluster transformation
$R=\theta_i(u_{>0}(t''), u_{<0}(t''))$.

For $q=t'$ or $t''$ introduce $\tau_q=\prod_\ell x_\ell(q)^{\beta_{i\ell}(q)}=\frac{u_{>0}(q)}{u_{<0}(q)}$.
%Note that $\rho(u_{>0},u_{<0})=\rho(\tau_q)\cdot L$, where $L$ is a Laurent monomial.

Finally,
\begin{multline*}
    \tau_{t''}|_{x_j\leftarrow \frac{Q_0}{x_j}}=\prod_{l\ne j} x_l^{\beta'_{il}}\cdot
\left(\frac{Q_0}{x_j}\right)^{-\beta'_{ij}}\prod_{\beta_{jl}(t')<0} x_l^{-b'_{ij}b'_{jl}/d_i}=\\
=\prod_l x_l^{\beta_{il}(t)}\left(p_{j;0}\prod_{b'_{jl}<0} x_l^{-b'_{jl}}\right)^{-b'_{ij}/d_i}
\left(\prod_{b'_{jl}<0} x_l^{-b'_{ij} b'_{jl}/d_i}\right)=p_{j;0}^{-\beta'_{ij}}\tau_{t'}=\frac{1}{p_{j;0}^{\beta'_{ij}}\tau_{t}}
\end{multline*}

It is enough to notice that \ref{eq:GenClusterMutation} and \ref{eq:GenCoeffMutation} imply
that $\rho_{t;i}(\tau_{t})=\rho_{t'';i}\left(\tau_{t''}|_{x_j\leftarrow \frac{Q_0}{x_j}}\right)\cdot L$, 
where $L$ is a Laurent monomial.

Therefore, $R|_{x_j\leftarrow \frac{Q_0}{x_j}}=P\cdot\hat L$ where $\hat L$ is another Laurent monomial..
\end{proof}

\begin{theorem}\label{thm:generalizedCartanKilling}
Generalized cluster algebras of finite type satisfy the same Cartan-Killing classification as the standard cluster algebras.
\end{theorem}

\begin{proof}
The proof repeats the one of~\cite{FZ2}. The only differences make the proofs of the fact that the only finite type generalized cluster algebras of rank two correspond to $A_2, B_2, G_2$ types. It is checked by direct computation similar to one in \cite{FZ2}. Note first that in $A_2$-case formulas for generalized cluster transformation coincide with formulas for the standard cluster transformation. In $B_2$-case the polynomial
degrees of theta-polynomials are two and one.  Set the theta polynomials in the initial cluster  $\theta_1(u,v)=a u^2+b u v+c v^2$, $\theta_2(u,v)=p u+q v$.
Then, we immediately obtain $(x,y)\overset{\mu_1}\longleftrightarrow(\mu_1)(x_1,y)\overset{\mu_2}\longleftrightarrow (x_1,y_1)\overset{\mu_1}\longleftrightarrow (x_2,y_1)\overset{\mu_2}\longleftrightarrow
(x_2,y_2)\overset{\mu_1}\longleftrightarrow (x,y_2)\overset{\mu_2}\longleftrightarrow (x,y)$, where

$x_1=(a+by+cy^2)/x$,

$y_1=(px+qa+bqy+cqy^2)/xy$,

$x_2=(a^2q^2+2ap q x+acq^2 y^2+abq^2 y+bpq  x  y+p^2  x^2)/x  y^2$,

$y_2=(q  a+p   x)/y$.

Similar computations lead to the 8-cycle in $G_2$-case. Note that degrees of polynomials $\deg(\theta_1)=3$, $\deg(\theta_2)=1$. We set
$\theta_1(u,v)=a u^3+b u^2 v+c u v^2+ d v^3$, $\theta_2(u,v)=p u+q v$.
$(x,y)\overset{\mu_1}\longleftrightarrow(\mu_1)(x_1,y)\overset{\mu_2}\longleftrightarrow (x_1,y_1)\overset{\mu_1}\longleftrightarrow (x_2,y_1)\overset{\mu_2}\longleftrightarrow
(x_2,y_2)\overset{\mu_1}\longleftrightarrow (x_3,y_2)\overset{\mu_2}\longleftrightarrow (x_3,y_3)\overset{\mu_1}\longleftrightarrow (x,y_3)\overset{\mu_2}\longleftrightarrow (x,y)$, where

$x_1=(a+b  y+c  y^2+d  y^3)/x$,

$y_1=(p  x+aq +bq  y+cq  y^2+dq  y^3)/x  y$,

$x_2=(a^3q^3 +2 a^2 cq^3  y^2+2a^2 d q^3 y^3+3 a^2 p  q^2  x+2a^2 b q^3  y+2a b d q^3 y^4 +3 a p^2 q x^2
+4 a  b p   q^2  x  y+a c^2 q^3  y^4+2acd  q^3  y^5 +ad^2q^3  y^6+3ac  p  q^2  x  y^2+3ad  p   q^2  x  y^3+
a b^2 q^3   y^2
+2 a  b c q^3  y^3 +bc p   q^2  x  y^3 +bd p   q^2  x  y^4 +
p^3  x^3+b^2 p   q^2 x  y^2+2 b p^2 q x^2  y+p^2 cq x^2  y^2)/x^2  y^3$,

$y_2=(q^2  a^2+abq^2 y+acq^2 y^2+adq^2 y^3+2 a p   q  x+bp  q  x  y+p^2  x^2)/xy^2 $,

$x_3=(  a^2  d q^3 y^3+  a^2  c q^3 y^2+acp  q^2  x  y^2+  a^2  b q^3  y+2 ab p   q^2  x  y+bp^2 q x^2  y+
a^3 q^3 +3 a^2 p  q^2  x+3  a p^2 q x^2 +p^3  x^3)/x y^3$,

$y_3=(p  x+a q )/y$.

\end{proof}

We note that according to~\cite{GSV} the generalized cluster transformation preserves presymplectic structure compatible with cluster algebra structure. Similarly, the secondary generalized cluster transformation preserves
the compatible Poisson bracket.

%As for the standard cluster algebras we can consider modified $y$-variables $\hat y_i=y_i\prod_{j\ne i} x_j^{b_{ij}}$
%(cfg, also with $\tau$-coordinates of\cite{GSV1,GSVbook}). Then we observe that
%$\bar y$ satisfy so-called $y$-dynamics. Namely, under the generalized cluster transformation $(\hat y'_i)=\mu'_k(\hat y_i)$ one has $\hat y'_k=1/\hat y_k$, and $\hat y'_j=\hat y_j \hat y_k^{[b_{jk}]_+}\cdot \rho(\hat y_k)^{-\beta_{jk}}$, for $j\ne k$.

\begin{theorem}\label{thm:generalizedCompatible}(\cite{GSV}) Poisson structure compatible with a cluster algebra is compatible with the corresponding generalized cluster transformations.
\end{theorem}

%The following positivity conjecture is a generalization of the positivity conjecture for cluster algebra.
%
%\begin{hypothesis}\label{generalizedPositivity} If $\rho$ is a reciprocal polynomial with positive coefficients then any cluster variable  of a generalized cluster algebra is expressed as a positive Laurent polynomial in the initial cluster.
%\end{hypothesis}
%
%This conjecture is checked for finite type rank 2 cluster algebras.

In the next section we describe generalized cluster structure associated with triangulated surfaces with orbifold points.
%
%Only the reciprocal polynomial of second order of type $\rho(z)=1+2\cos(\pi/n)z+z^2$ occur in this geometric situation.
%In this case we know the positivity by the following arguments. The generalized cluster algebra constructed in this way
%is a subalgebra of a standard cluster algebra associated with another triangulated surface where generalized exchange relation
%can be obtained as a sequence of standard mutations. Therefore the positivity conjecture follows from the known result on positivity
%of Laurent polynomials representing cluster variables associated with an arc of triangulated surface (see~\cite{MSchW}).
%
%This example leads to a natural question:
%
%{\bf Question: } \emph{Is any generalized cluster algebra  a subalgebra of some standard cluster algebra?}
%

\section{Teichm\"uller space of surfaces with holes and orbifold points of arbitrary order}\label{s:geometry}

We now demonstrate how the above mutations with reciprocal polynomials of the second order appear in the description of
Teichm\"uller spaces of Riemann surfaces of arbitrary genus with nonzero number of holes (punctures) and with
an arbitrary number $r$ of orbifold points of arbitrary orders.

\subsection{The ideal triangle decompositions of orbifold Riemann surfaces and cluster variables}\label{ss:clusters}

We now present the geometric
pattern underlying the algebraic construction of cluster variables corresponding to
orbifold Riemann surfaces.

Particular cases of Riemann surfaces with orbifold points of order 2 and 3 are discussed in~\cite{Ch1,Ch2}, \and\cite{ChM}.
For relation between skew-symmetrizable cluster algebras of finite mutation type and
surfaces with orbifold points see \cite{FST-Orb}.

We consider a regular Riemann surface $\Sigma_{g,s,r}$
 of genus $g$ with $s>0$ holes and with a number $r\ge 0$ of orbifold points; orders of
these points $p_i$, $i=1,\dots,r$ are positive integers greater than one.

We introduce the marking on the set of orbifold points splitting this set into nonintersecting
(possibly empty) subsets $\delta_k$,
$k=1,\dots,s$, $\sum_{k=1}^s|\delta_k|=r$. For every $k$,
we then assign the subset $\delta_k$ to the $k$th hole and introduce a cyclic ordering
inside each subset $\delta_k$.

To construct the generalization of the ideal triangle decomposition \cite{Fock1,Penn1}, we first remove
from the surface all the hyperbolic domains of holes bounded by the corresponding
perimeter geodesic lines (with their closures, which are these perimeter lines).
Second, we choose for each orbifold point from the set $\delta_k$ a domain containing this point and bounded by a geodesic curve whose both ends spiral to the $k$th hole as shown in the right part of
Fig.~\ref{fi:saucer-pan}.\footnote{In the case of a $\mathbb Z_2$-orbifold point, this domain has zero area because the corresponding
geodesics goes straight to the orbifold point, reflects at it, and repeats its path in the opposite direction.}
We remove from the remaining part of the surface all such domains.
The remaining part of $\Sigma_{g,s,r}$ admits splitting into ideal triangles; a copy of this
splitting can be drawn as a connected ideal polygon in the Poincar\'e disc; the sides of
this polygon are of two sorts: those that are not pre-images of geodesics going around orbifold points
must be pairwise identified; to a side that is a pre-image of the geodesics going around ${\mathbb Z}_p$
orbifold point we attach (from outside) an equilateral ideal $p$-gone with the orbifold point situated at
its geodesic center; this $p$-gone is the $p$-fold covering of the removed domain enclosing the orbifold point.

An example of a fundamental domain of the Riemann surface $\Sigma_{1,1,2}$ of genus one with one hole and
with two ${\mathbb Z}_3$ orbifold points is in Fig.~\ref{fi:disc}.

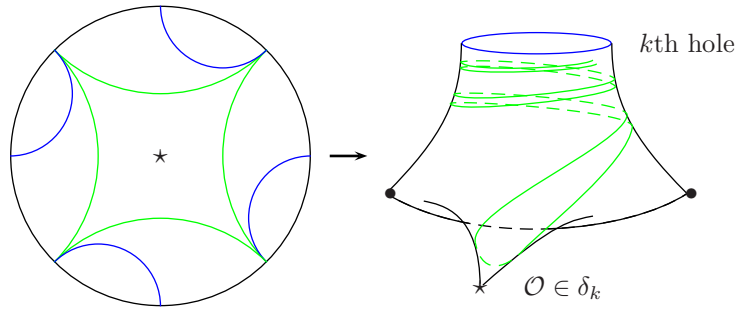
\begin{figure}[tb]
%\hspace*{2cm}
%\epsfysize=6cm
%\vskip .2in
{\psset{unit=0.5}
\begin{pspicture}(-9,1)(9,-7)
%First half of the figure
\rput(-5,-3){
\pscircle[linewidth=0.5pt](0,0){4}
\rput(0,0){\psarc[linecolor=green, linewidth=0.5pt](5.657,0){4}{135}{225}}
\rput(0,0){\psarc[linecolor=blue, linewidth=0.5pt](4,-1.657){1.657}{90}{225}}
\rput{90}(0,0){\psarc[linecolor=green, linewidth=0.5pt](5.657,0){4}{135}{225}}
\rput{90}(0,0){\psarc[linecolor=blue, linewidth=0.5pt](4,-1.657){1.657}{90}{225}}
\rput{180}(0,0){\psarc[linecolor=green, linewidth=0.5pt](5.657,0){4}{135}{225}}
\rput{180}(0,0){\psarc[linecolor=blue, linewidth=0.5pt](4,-1.657){1.657}{90}{225}}
\rput{270}(0,0){\psarc[linecolor=green, linewidth=0.5pt](5.657,0){4}{135}{225}}
\rput{270}(0,0){\psarc[linecolor=blue, linewidth=0.5pt](4,-1.657){1.657}{90}{225}}
\rput(0,0){\makebox(0,0){\large $\star$}}
}
\rput(0,-3){\pcline[linewidth=1pt]{->}(-0.5,0)(0.5,0)}
%Second half of the figure
\rput(5,0){
%bottleneck curve
\psellipse[linecolor=blue, linewidth=0.5pt](0,0)(2,0.3)
\rput(4,0){\makebox(0,0){$k$th hole}}
%bounding curves
\psbezier[linewidth=0.5pt](-4,-4)(-2.5,-2.5)(-2,-1.5)(-2,0)
\psbezier[linewidth=0.5pt](4,-4)(2.5,-2.5)(2,-1.5)(2,0)
\psbezier[linewidth=0.5pt](-4,-4)(-1.5,-5.5)(2.5,-5)(4,-4)
\psframe[linecolor=white, fillstyle=solid, fillcolor=white](-1.8,-5.3)(0.3,-4)
\psbezier[linestyle=dashed, linewidth=0.5pt](-4,-4)(-1.5,-5.5)(2.5,-5)(4,-4)
\psbezier[linewidth=0.5pt](-3,-4.2)(-2.5,-4.3)(-1.5,-4.5)(-1.5,-6.5)
\psbezier[linewidth=0.5pt](1.5,-4.6)(1,-4.7)(0.5,-4.7)(-1.5,-6.5)
\pscircle[linecolor=white, fillstyle=solid, fillcolor=white](-1.5,-6.5){0.1}
%winding curves
%0th coil
\psbezier[linecolor=green, linestyle=dashed, linewidth=0.5pt](-0.35,-5.5)(-0.8,-6)(-1.46,-6.2)(-1.63,-5.3)
\psbezier[linecolor=green, linewidth=0.5pt](-0.35,-5.5)(0.1,-5)(2.45,-3.2)(2.55,-2.2)
\psbezier[linecolor=green, linewidth=0.5pt](-1.63,-5.3)(-1.8,-4.4)(2.35,-2.7)(2.41,-1.9)
%\psbezier[linecolor=green, linewidth=0.5pt](-1.5,-6.5)(0,-3.5)(2.4,-2.8)(2.55,-2.2)
%\psbezier[linecolor=green, linewidth=0.5pt](-4,-4)(-1,-2.5)(2.3,-2.4)(2.41,-1.9)
%\psbezier[linecolor=green, linestyle=dashed, linewidth=0.5pt](4,-4)(1,-2.4)(-2.2,-2.1)(-2.55,-2.2)
%1st coil
\psbezier[linecolor=green, linestyle=dashed, linewidth=0.5pt](-2.3,-1.6)(-2,-1.5)(2.4,-1.8)(2.55,-2.2)
\psbezier[linecolor=green, linestyle=dashed, linewidth=0.5pt](-2.2,-1.35)(-1.8,-1.25)(2.3,-1.5)(2.41,-1.9)
%2nd coil
%\psbezier[linecolor=green, linewidth=0.5pt](-2.55,-2.2)(-2.45,-2.5)(2.1,-1.8)(2.27,-1.55)
\psbezier[linecolor=green, linewidth=0.5pt](-2.3,-1.6)(-2.2,-1.9)(2,-1.45)(2.1,-1.15)
\psbezier[linecolor=green, linewidth=0.5pt](-2.2,-1.35)(-2.1,-1.65)(2,-1.25)(2.08,-.95)
%3rd coil
%\psbezier[linecolor=green, linestyle=dashed, linewidth=0.5pt](-2.09,-1)(-1.45,-.85)(2.05,-1.25)(2.27,-1.55)
\psbezier[linecolor=green, linestyle=dashed, linewidth=0.5pt](-2.03,-.65)(-1.4,-0.5)(1.9,-.85)(2.1,-1.15)
\psbezier[linecolor=green, linestyle=dashed, linewidth=0.5pt](-2.03,-.5)(-1.4,-.37)(1.86,-.65)(2.08,-.95)
%4th coil
%\psbezier[linecolor=green, linewidth=0.5pt](-2.09,-1)(-2,-1.25)(1.6,-1.05)(1.9,-.7)
\psbezier[linecolor=green, linewidth=0.5pt](-2.03,-.65)(-1.95,-0.87)(1.3,-.75)(1.6,-.55)
\psbezier[linecolor=green, linewidth=0.5pt](-2.01,-.5)(-1.92,-0.78)(1.2,-.67)(1.5,-.45)
%orbifold points
\rput(-4,-4){\makebox(0,0){ $\bullet$}}
\rput(4,-4){\makebox(0,0){ $\bullet$}}
\rput(-1.5,-6.5){\makebox(0,0){\large $\star$}}
\rput(0.7,-6.5){\makebox(0,0){${\mathcal O}\in\delta_k$}}
}
\end{pspicture}
}
\caption{\small The right part represents an example of a regular genus zero
Riemann surface contaning a ${\mathbb Z}_4$-orbifold point ${\mathcal O}$ marked by a star; the bounding
geodesic line is the image of a side of an
ideal equilateral square in the Poincar\'e disc depicted in the left part.
Both ends of the geodesic line spiral asymptotically to the
closed geodesic that is the boundary of the $k$th hole.}
\label{fi:saucer-pan}
\end{figure}

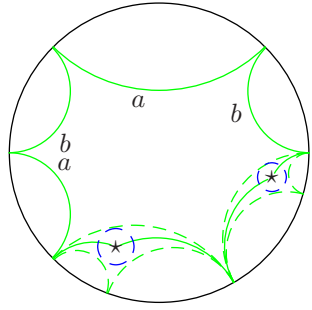
\begin{figure}[tb]
%\hspace*{2cm}
%\epsfysize=6cm
%\vskip .2in
{\psset{unit=0.5}
\begin{pspicture}(-10,-5)(8,5)
%rotated closed geodesics
\pscircle[linewidth=0.5pt](0,0){4}
%curves, bounding the fundamental region
\psarc[linecolor=green, linewidth=0.5pt](-4,1.64){1.64}{-90}{45}
\psarc[linecolor=green, linewidth=0.5pt](0,5.66){4}{-135}{-45}
\psarc[linecolor=green, linewidth=0.5pt](4,1.64){1.64}{135}{270}
\psarc[linestyle=dashed,linecolor=green, linewidth=0.5pt](4,-2.31){2.31}{90}{210}
\rput{142}(0,0){\psarc[linestyle=dashed, linecolor=green, linewidth=0.5pt](-4.314,0){1.616}{-68}{68}}
\rput{172}(0,0){\psarc[linestyle=dashed, linecolor=green, linewidth=0.5pt](-4.04,0){.56}{-82}{82}}
\rput{146}(0,0){\psarc[linecolor=green, linewidth=0.5pt](-4.45,0){1.95}{-32}{64}}
\rput{165}(0,0){\psarc[linecolor=green, linewidth=0.5pt](-4.14,0){1.07}{-75}{-14}}
\psarc[linestyle=dashed, linecolor=green, linewidth=0.5pt](-.67,-4.99){3.06}{30}{135}
\rput{57.5}(0,0){\psarc[linestyle=dashed, linecolor=green, linewidth=0.5pt](-4.09,0){.85}{-78}{78}}
\rput{95}(0,0){\psarc[linestyle=dashed, linecolor=green, linewidth=0.5pt](-4.40,0){1.86}{-65}{65}}
\rput{66}(0,0){\psarc[linecolor=green, linewidth=0.5pt](-4.285,0){1.535}{6}{69}}
\rput{89}(0,0){\psarc[linecolor=green, linewidth=0.5pt](-4.666,0){2.4}{-59}{25}}
\psarc[linecolor=green, linewidth=0.5pt](-4,-1.64){1.64}{-45}{90}
%\rput(-2.75,0.6){\makebox(0,0){\tiny $\bullet$}}
%\rput(-0.55,1.7){\makebox(0,0){\tiny $\bullet$}}
%\rput(2.35,1.45){\makebox(0,0){\tiny $\bullet$}}
\rput(3,-0.65){\makebox(0,0){$\star$}}
\pscircle[linestyle=dashed, linecolor=blue, linewidth=0.5pt](3,-.65){0.4}
\rput(-1.15,-2.5){\makebox(0,0){$\star$}}
\pscircle[linestyle=dashed, linecolor=blue, linewidth=0.5pt](-1.15,-2.5){0.5}
%\rput(-2.8,-0.45){\makebox(0,0){\tiny $\bullet$}}
% marking
\rput(-2.65,0.5){\makebox(0,0)[lt]{$b$}}
\rput(-2.7,-0.45){\makebox(0,0)[lb]{$a$}}
%\rput(-.7,-3.4){\makebox(0,0)[lb]{$s_3$}}
%\rput(3.2,-1.7){\makebox(0,0)[rb]{$s_4$}}
\rput(2.2,1.3){\makebox(0,0)[rt]{$b$}}
\rput(-0.55,1.5){\makebox(0,0)[ct]{$a$}}
\end{pspicture}
}
\caption{\small
The Poincar\'e disc with depicted fundamental domain for the genus-one surface $\Sigma_{1,1,2}$ with
one hole and two ${\mathbb Z}_3$ orbifold points (marked by $\star$). Solid lines
constitute the boundary of a fundamental domain and dashed lines are sides of the related ideal triangles.
The sides with labels $a$ and $b$ are pairwise identified.
Inscribed circles with centers at the ${\mathbb Z}_3$ orbifold points all have the radius $\frac12\log3$.}
\label{fi:disc}
\end{figure}

\begin{figure}[tb]
{\psset{unit=0.5}
%{\psset{unit=0.7}
\begin{pspicture}(6,-5)(0,5)
%rotated closed geodesics
\pscircle[linewidth=0.5pt](0,0){4}
%curves, bounding the fundamental region
\newcommand{\PATTERN}{%
\psarc[linecolor=green, linewidth=1pt](4.62,0){2.3}{120}{240}
\rput{10}(0,0){\rput(2.3,0){\makebox(0,0)[rc]{$c$}}}
\rput{-20}(0,0){\psarc[linestyle=dashed, linewidth=1pt](4.06,0){.7}{100}{260}}
\rput{10}(0,0){\psarc[linestyle=dashed, linewidth=1pt](4.26,0){1.45}{110}{250}}
\rput{-20}(0,0){\rput(3.9,0){\makebox(0,0)[rc]{$a$}}}
\rput{10}(0,0){\rput(3.8,0){\makebox(0,0)[rc]{$b$}}}
}
\rput{30}(0,0){\PATTERN}
\rput{90}(0,0){\PATTERN}
\rput{150}(0,0){\PATTERN}
\rput{210}(0,0){\PATTERN}
\rput{270}(0,0){\PATTERN}
\rput{330}(0,0){\PATTERN}
\rput{60}(0,0){
\psarc[linestyle=dashed,linecolor=green, linewidth=1pt](8,0){6.93}{150}{210}
}
\rput{-60}(0,0){
\psarc[linestyle=dashed,linecolor=green, linewidth=1pt](8,0){6.93}{150}{210}
}
\pcline[linestyle=dashed,linecolor=green, linewidth=1pt](-4,0)(4,0)
\pcline[linewidth=2pt]{->}(6,0)(7,0)
\rput(0,0){\makebox(0,0){$\star$}}
\rput(-0.6,0.1){\makebox(0,0)[cb]{$c_3$}}
\rput(1,-0.9){\makebox(0,0)[ct]{$c_2$}}
\rput(.4,1.2){\makebox(0,0)[cb]{$c_2$}}
%\rput(-1,0.1){\makebox(0,0)[cb]{$Y_3$}}
%\rput(1,-0.6){\makebox(0,0)[ct]{$Y_2$}}
%\rput(1,0.6){\makebox(0,0)[cb]{$Y_4$}}
%\rput(1.8,-1.3){\makebox(0,0)[lt]{$Z_1$}}
%\rput(-0.4,-2.2){\makebox(0,0)[lt]{$Z_2$}}
%\rput(-1.8,-1.3){\makebox(0,0)[rt]{$Z_3$}}
%\rput(1.8,1.3){\makebox(0,0)[lb]{$Z_6$}}
%\rput(-0.4,2.2){\makebox(0,0)[lb]{$Z_5$}}
%\rput(-1.8,1.3){\makebox(0,0)[rb]{$Z_4$}}
%\rput(3.4,-2.2){\makebox(0,0)[lt]{$e^{i\phi}$}}
%\rput(2,-3.5){\makebox(0,0)[lt]{$e^{2i\pi/p}$}}
\end{pspicture}
%
% second half
%
\begin{pspicture}(-6,-5)(0,5)
%rotated closed geodesics
\pscircle[linewidth=0.5pt](0,0){4}
%curves, bounding the fundamental region
\newcommand{\PATTERN}{%
\rput{20}(0,0){\psarc[linecolor=green, linewidth=1pt](4.62,0){2.3}{120}{240}}
\rput{30}(0,0){\rput(2.2,-0.3){\makebox(0,0)[rc]{$c'$}}}
\rput{-20}(0,0){\psarc[linestyle=dashed, linewidth=1pt](4.06,0){.7}{100}{260}}
\rput{10}(0,0){\psarc[linestyle=dashed, linewidth=1pt](4.26,0){1.45}{110}{250}}
\rput{-20}(0,0){\rput(3.9,0){\makebox(0,0)[rc]{$a$}}}
\rput{10}(0,0){\rput(3.8,0){\makebox(0,0)[rc]{$b$}}}
}
\rput{30}(0,0){\PATTERN}
\rput{90}(0,0){\PATTERN}
\rput{150}(0,0){\PATTERN}
\rput{210}(0,0){\PATTERN}
\rput{270}(0,0){\PATTERN}
\rput{330}(0,0){\PATTERN}
\rput{80}(0,0){
\psarc[linestyle=dashed,linecolor=green, linewidth=1pt](8,0){6.93}{150}{210}
}
\rput{-40}(0,0){
\psarc[linestyle=dashed,linecolor=green, linewidth=1pt](8,0){6.93}{150}{210}
}
\rput{20}(0,0){
\pcline[linestyle=dashed,linecolor=green, linewidth=1pt](-4,0)(4,0)
}
\rput(0,0){\makebox(0,0){$\star$}}
\rput(-0.6,0.1){\makebox(0,0)[cb]{$c'_3$}}
\rput(0.3,-0.2){\makebox(0,0)[ct]{$c'_2$}}
\rput(.4,1.2){\makebox(0,0)[cb]{$c'_2$}}
%\rput(-1,0.1){\makebox(0,0)[cb]{$Y_3$}}
%\rput(1,-0.6){\makebox(0,0)[ct]{$Y_2$}}
%\rput(1,0.6){\makebox(0,0)[cb]{$Y_4$}}
%\rput(1.8,-1.3){\makebox(0,0)[lt]{$Z_1$}}
%\rput(-0.4,-2.2){\makebox(0,0)[lt]{$Z_2$}}
%\rput(-1.8,-1.3){\makebox(0,0)[rt]{$Z_3$}}
%\rput(1.8,1.3){\makebox(0,0)[lb]{$Z_6$}}
%\rput(-0.4,2.2){\makebox(0,0)[lb]{$Z_5$}}
%\rput(-1.8,1.3){\makebox(0,0)[rb]{$Z_4$}}
%\rput(3.4,-2.2){\makebox(0,0)[lt]{$e^{i\phi}$}}
%\rput(2,-3.5){\makebox(0,0)[lt]{$e^{2i\pi/p}$}}
\end{pspicture}
}
\caption{\small
The mutation in the ideal $p$-gone (here $p=6$), which is the $p$-fold covering
of the domain around a $\mathbb Z_p$-orbifold point (marked by $\star$). We must perform
a sequence of standard
$2$-term mutations (\ref{diagonal}) on the set of cluster variables to come from the
pattern in the left-hand side to the one in the right-hand side.}
\label{fi:p-gone}
\end{figure}
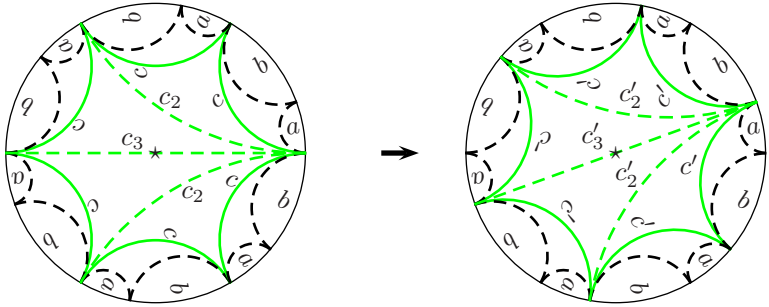

We use the standard geometric correspondence between cluster variables and $\lambda$-lengths: at each point at the
absolute that is a vertex of an ideal triangle we set an horocycle; the $\lambda$-length $\ell_a$ is then
the (signed) geodesic length of the part of side $a$ enclosed between two horocycles based at its endpoints,
or the signed distance between horocycles: $\ell_a$ is negative when the corresponding horocycles
overlap. The correspondence reads
\be
\label{lambda-l}
a=e^{\ell_a/2},
\ee
and $a$ is the cluster variable associated with the edge.

For an ideal quadrangle with the (cyclically enumerated) sides
$a_i$, $i=1,\dots,4$ and with diagonals $d$ and $d'$ we have the celebrated two-term cluster relation
\be
\label{diagonal}
dd'=a_1a_3+a_2a_4,
\ee
which holds independently of the choice of horocycles.

We now consider mutations for cluster variables of the ideal $p$-gone corresponding to a $\mathbb Z_p$-orbifold point.
We consider the pattern in the left-hand side of Fig.~\ref{fi:p-gone} and perform a sequence of mutations (\ref{diagonal})
to come to the pattern in the right-hand side. For an equilateral $p$-gone, the cluster variable $c_k$ for a $k$-diagonal is
\be
\label{ck}
c_k=c\frac{\sin(\pi k/p)}{\sin(\pi/p)},\qquad (c_1= c).
\ee
The easy combinatorics then yields
\begin{lm}\label{cc-prime} $\lambda$-lengths satisfy the following relation
$$
cc'=a^2+2\cos(\pi/p)ab+b^2.
$$
\end{lm}

\begin{remark} Relation in Lemma~\ref{cc-prime} is a generalized cluster mutation described in Sec.~\ref{s-algebra}.
Note that since we use only two-term transformations \ref{diagonal} to prove Lemma~\ref{cc-prime}, the positivity property for
generalized mutations of such form follows from the one for the standard mutations.
\end{remark}

To simplify the description, it is convenient to introduce the notion of {\em petal surface}. Petals are the domains containing orbifold points that were removed on the second step of constructing the ideal triangulation.  To the
$k$th hole, we associate the bouquet of $|\delta_k|$ petals (with no petals if $\delta_k$ is empty), each petal carries,
besides its cluster variable, the
number $\omega_p=2\cos(\pi/p)$. The mutation then occurs inside the corresponding ideal triangle (painted by a light color), and we have three cases
depending on whether the adjacent sides are petals themselves:
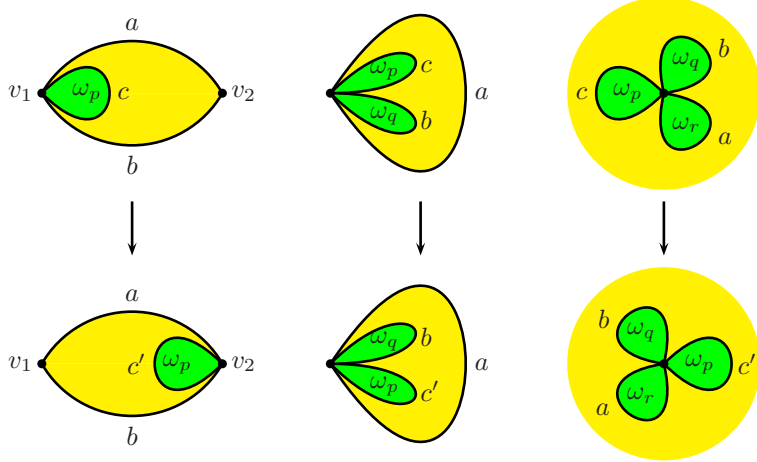
\begin{center}
{\psset{unit=1.2}
\begin{pspicture}(-1.5,-2.5)(1.5,2.5)
\psarc[linewidth=1pt,fillstyle=solid,fillcolor=yellow](0,.923){1.155}{30}{150}
\psarc[linewidth=1pt,fillstyle=solid,fillcolor=yellow](0,2.077){1.155}{210}{330}
\psbezier[linewidth=1pt,fillstyle=solid,fillcolor=green](-1,1.5)(0,2.5)(0,0.5)(-1,1.5)
\pscircle*(-1,1.5){0.05}
\pscircle*(1,1.5){0.05}
\pcline[linewidth=1pt]{->}(0,0.3)(0,-0.3)
\psarc[linewidth=1pt,fillstyle=solid,fillcolor=yellow](0,-.923){1.155}{210}{330}
\psarc[linewidth=1pt,fillstyle=solid,fillcolor=yellow](0,-2.077){1.155}{30}{150}
\psbezier[linewidth=1pt,fillstyle=solid,fillcolor=green](1,-1.5)(0,-2.5)(0,-0.5)(1,-1.5)
\pscircle*(-1,-1.5){0.05}
\pscircle*(1,-1.5){0.05}
\rput(-0.5,1.5){\makebox(0,0)[cc]{$\omega_p$}}
\rput(-0.1,1.5){\makebox(0,0)[cc]{$c$}}
\rput(0,2.2){\makebox(0,0)[cb]{$a$}}
\rput(0,0.8){\makebox(0,0)[ct]{$b$}}
\rput(-1.1,1.5){\makebox(0,0)[rc]{$v_1$}}
\rput(1.1,1.5){\makebox(0,0)[lc]{$v_2$}}
\rput(0.5,-1.5){\makebox(0,0)[cc]{$\omega_p$}}
\rput(0.05,-1.5){\makebox(0,0)[cc]{$c'$}}
\rput(0,-2.2){\makebox(0,0)[ct]{$b$}}
\rput(0,-0.8){\makebox(0,0)[cb]{$a$}}
\rput(-1.1,-1.5){\makebox(0,0)[rc]{$v_1$}}
\rput(1.1,-1.5){\makebox(0,0)[lc]{$v_2$}}
\end{pspicture}
}
{\psset{unit=1.2}
\begin{pspicture}(-1.5,-2.5)(1,2.5)
\psbezier[linewidth=1pt,fillstyle=solid,fillcolor=yellow](-1,1.5)(1,4.5)(1,-1.5)(-1,1.5)
\psbezier[linewidth=1pt,fillstyle=solid,fillcolor=green](-1,1.5)(0,2.5)(.5,1.5)(-1,1.5)
\psbezier[linewidth=1pt,fillstyle=solid,fillcolor=green](-1,1.5)(0,0.5)(.5,1.5)(-1,1.5)
\pscircle*(-1,1.5){0.05}
\pcline[linewidth=1pt]{->}(0,0.3)(0,-0.3)
\psbezier[linewidth=1pt,fillstyle=solid,fillcolor=yellow](-1,-1.5)(1,-4.5)(1,1.5)(-1,-1.5)
\psbezier[linewidth=1pt,fillstyle=solid,fillcolor=green](-1,-1.5)(0,-2.5)(.5,-1.5)(-1,-1.5)
\psbezier[linewidth=1pt,fillstyle=solid,fillcolor=green](-1,-1.5)(0,-0.5)(.5,-1.5)(-1,-1.5)
\pscircle*(-1,-1.5){0.05}
\rput(-0.4,1.75){\makebox(0,0)[cc]{$\omega_p$}}
\rput(-0.4,1.25){\makebox(0,0)[cc]{$\omega_q$}}
\rput(0,1.8){\makebox(0,0)[lc]{$c$}}
\rput(0.6,1.5){\makebox(0,0)[lc]{$a$}}
\rput(0,1.2){\makebox(0,0)[lc]{$b$}}
\rput(-0.4,-1.75){\makebox(0,0)[cc]{$\omega_p$}}
\rput(-0.4,-1.25){\makebox(0,0)[cc]{$\omega_q$}}
\rput(0,-1.8){\makebox(0,0)[lc]{$c'$}}
\rput(0,-1.2){\makebox(0,0)[lc]{$b$}}
\rput(0.6,-1.5){\makebox(0,0)[lc]{$a$}}
\end{pspicture}
}
{\psset{unit=1.2}
\begin{pspicture}(-1.5,-2.5)(1.5,2.5)
\pscircle[linewidth=1pt,linecolor=white,fillstyle=solid,fillcolor=yellow](0,1.5){1.1}
\psbezier[linewidth=1pt,fillstyle=solid,fillcolor=green](0,1.5)(-0.3,2.8)(1.3,1.7)(0,1.5)
\psbezier[linewidth=1pt,fillstyle=solid,fillcolor=green](0,1.5)(-0.3,.2)(1.3,1.3)(0,1.5)
\psbezier[linewidth=1pt,fillstyle=solid,fillcolor=green](0,1.5)(-1,2.5)(-1,.5)(0,1.5)
\pscircle*(0,1.5){0.05}
\pcline[linewidth=1pt]{->}(0,0.3)(0,-0.3)
\pscircle[linewidth=1pt,linecolor=white,fillstyle=solid,fillcolor=yellow](0,-1.5){1.1}
\psbezier[linewidth=1pt,fillstyle=solid,fillcolor=green](0,-1.5)(0.3,-2.8)(-1.3,-1.7)(0,-1.5)
\psbezier[linewidth=1pt,fillstyle=solid,fillcolor=green](0,-1.5)(0.3,-.2)(-1.3,-1.3)(0,-1.5)
\psbezier[linewidth=1pt,fillstyle=solid,fillcolor=green](0,-1.5)(1,-2.5)(1,-.5)(0,-1.5)
\pscircle*(0,-1.5){0.05}
\rput(-0.4,1.5){\makebox(0,0)[cc]{$\omega_p$}}
\rput(0.25,1.85){\makebox(0,0)[cc]{$\omega_q$}}
\rput(0.25,1.15){\makebox(0,0)[cc]{$\omega_r$}}
\rput(-0.85,1.5){\makebox(0,0)[rc]{$c$}}
\rput(0.6,1){\makebox(0,0)[lc]{$a$}}
\rput(0.6,2){\makebox(0,0)[lc]{$b$}}
%second set of variables
\rput(0.4,-1.5){\makebox(0,0)[cc]{$\omega_p$}}
\rput(-0.25,-1.85){\makebox(0,0)[cc]{$\omega_r$}}
\rput(-0.25,-1.15){\makebox(0,0)[cc]{$\omega_q$}}
\rput(0.82,-1.5){\makebox(0,0)[lc]{$c'$}}
\rput(-0.6,-1){\makebox(0,0)[rc]{$b$}}
\rput(-0.6,-2){\makebox(0,0)[rc]{$a$}}
\end{pspicture}
}
\end{center}
In all the three cases above, the
mutation law is given by Lemma~\ref{cc-prime}. In the first case, we transfer the cluster variable
from one set $\delta_k$ to another set (if the vertices $v_1$ and $v_2$ are distinct); these transformations may also
change the cyclic ordering of orbifold points inside a set $\delta_k$.
%$$
%cc'=a^2+\omega_p ab+b^2.
%$$
Note that the label $\omega_p$, being a term of the coefficient tuple, remains assigned to
the transformed edge.

\subsection{Fat graph description for Riemann surfaces with holes and with ${\mathbb Z}_p$ orbifold points}

In this subsection and in the rest of the paper, we use the graphs dual to the above ideal triangle decompositions
of Riemann surfaces. These graphs are especially useful when describing the Fuchsian groups
$\Delta_{g,s,r}$ of Riemann
surfaces $\Sigma_{g,s,r}$ and the corresponding geodesic functions.

\begin{df}\label{def-pend}
{\rm
We call a fat graph (a graph with the prescribed cyclic ordering of edges
entering each vertex) $\Gamma_{g,s,r}$ a {\em spine of the Riemann surface} $\Sigma_{g,s,r}$
with $g$ handles, $s>0$ holes, and $r$ orbifold points of the corresponding orders $p_i$, $i=1,\dots,r$, if
\begin{itemize}
\item[(a)] this graph can be embedded  without self-intersections in $\Sigma_{g,s,r}$;
\item[(b)] all vertices of $\Gamma_{g,s,r}$ are three-valent except exactly $r$
one-valent vertices (endpoints of ``pending'' edges), which are placed at the corresponding
orbifold points;
\item[(c)] for an orbifold point from the set $\delta_k$, the corresponding pending edge
protrudes towards the interior of the face of the graph containing the $k$th hole; the cyclic
ordering of pending edges pointing towards the interior of this face coincide with that
of orbifold points in the set $\delta_k$;
\item[(d)] upon cutting along all edges of $\Gamma_{g,s,r}$ the Riemann surface
$\Sigma_{g,s,r}$ splits into $s$ polygons each containing exactly one hole and being
simply connected upon gluing this hole.
\end{itemize}
Edges of this graph are labeled by distinct integers $\alpha=1,2,\dots,6g-6+3s+2r$, and we set
into the correspondence the real number $Z_\alpha$ to each edge.
}
\end{df}

The first homotopy groups $\pi_1(\Sigma_{g,s,r})$ and $\pi_1(\Gamma_{g,s,r})$ coincide because
each closed path in $\Sigma_{g,s,r}$ can be homotopically transformed to a closed path in $\Gamma_{g,s,r}$
(taking into account paths that go around orbifold points)
in a unique way. The standard statement in hyperbolic geometry is that conjugacy classes of elements of
a Fuchsian group $\Delta_{g,s,r}$ are in the 1-1 correspondence with homotopy
classes of closed paths in the Riemann surface $\Sigma_{g,s.r}=\HH^2_+/\Delta_{g,s,r}$ and that the
``actual'' length $\ell_\gamma$
of a hyperbolic element $\gamma\in\Delta_{g,s,r}$ coincides with the minimum length of
curves from the corresponding homotopy class; it is then the length of a unique closed
geodesic line belonging to this class.

When orbifold points are present, the Fuchsian group contains besides hyperbolic elements also elliptic
elements corresponding to rotations about these orbifold points.
The corresponding generators ${\wtd F}_i$, $i=1,\dots,r$, of the rotations through $2\pi/p_i$
are conjugates of the matrices
\be
\label{F-p}
{\wtd F}_i=U_iF_{p_i}U_i^{-1},\qquad F_p=
\left(\begin{array}{cc} 0 & 1\\ -1 & -w\end{array}\right),\quad w=2\cos{\pi/p}, \quad p\ge 2.
\ee

The real numbers $Z_\alpha$ in Definition~\ref{def-pend} are the $h$-lengths \cite{Penn1}:
they are called the {\em (Thurston) shear
coordinates} \cite{ThSh},\cite{Bon2} in the case of punctured Riemann surface (without boundary components).
We preserve this notation and this term also in the case of orbifold surfaces.
These coordinates are related to the
cross-ratio relation for two adjacent ideal triangles constituting the ideal quadrangle with
the respective vertices (in the cyclic order) $a,b,c,d$ and diagonal $bd$. At the same time, they are
related to the cluster variables $a_i$ corresponding to the sides of the corresponding quadrangle.
We have
\be
\label{four-term}
e^{Z}=-\frac{(b-c)(d-a)}{(b-a)(d-c)}=\frac{a_1a_3}{a_2a_4},
\ee
and we obtain the parameter $Z_\alpha$ choosing $\{a,b,c,d\}=\{-1,0,e^{Z_\alpha},\infty\}$.

For example, the fat graph corresponding to the pattern in Fig.~\ref{fi:disc}  is depicted in Fig.~\ref{fi:treegraph}.

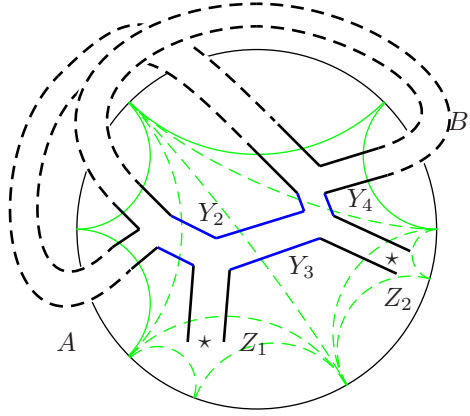
\begin{figure}[tb]
{\psset{unit=0.6}
\begin{pspicture}(-10,-5)(8,5)
%rotated closed geodesics
\pscircle[linewidth=0.5pt](0,0){4}
%curves, bounding the fundamental region
\psarc[linecolor=green, linewidth=0.5pt](-4,1.64){1.64}{-90}{45}
\psarc[linecolor=green, linewidth=0.5pt](0,5.66){4}{-135}{-45}
\psarc[linecolor=green, linewidth=0.5pt](4,1.64){1.64}{135}{270}
\psarc[linestyle=dashed,linecolor=green, linewidth=0.5pt](4,-2.31){2.31}{90}{210}
\rput{142}(0,0){\psarc[linestyle=dashed, linecolor=green, linewidth=0.5pt](-4.314,0){1.616}{-68}{68}}
\rput{172}(0,0){\psarc[linestyle=dashed, linecolor=green, linewidth=0.5pt](-4.04,0){.56}{-82}{82}}
\psarc[linestyle=dashed, linecolor=green, linewidth=0.5pt](-.67,-4.99){3.06}{30}{135}
\rput{57.5}(0,0){\psarc[linestyle=dashed, linecolor=green, linewidth=0.5pt](-4.09,0){.85}{-78}{78}}
\rput{95}(0,0){\psarc[linestyle=dashed, linecolor=green, linewidth=0.5pt](-4.40,0){1.86}{-65}{65}}
\psarc[linecolor=green, linewidth=0.5pt](-4,-1.64){1.64}{-45}{90}
%additional partition lines
\psbezier[linecolor=green, linestyle=dashed, linewidth=0.5pt](-2.83,2.83)(-1.4,1.4)(2,0)(4,0)
\psbezier[linecolor=green, linestyle=dashed, linewidth=0.5pt](-2.83,2.83)(-1.4,1.4)(1,-1.732)(2,-3.464)
\rput{90}(0,0){\psarc[linecolor=green, linestyle=dashed, linewidth=0.5pt](0,5.66){4}{-135}{-45}}
%lines of the graph
%1st pending edge
\pcline[linewidth=1pt](-3.1,0.5)(-2.6,0)
\pcline[linewidth=1pt](-2.6,1)(-1.9,0.3)
%2nd pending edge
\pcline[linewidth=1pt](-3.2,-0.4)(-2.6,0)
\pcline[linewidth=1pt](-2.8,-0.9)(-2.2,-0.4)
%3rd pending edge
\pcline[linewidth=1pt](-1.53,-2.5)(-1.4,-0.8)
\pcline[linewidth=1pt](-.73,-2.5)(-.6,-0.9)
%4th pending edge
\pcline[linewidth=1pt](1.7,.3)(3.4,-.485)
\pcline[linewidth=1pt](1.4,-.2)(3.1,-.985)
%5th pending edge
\pcline[linewidth=1pt](1.4,1.4)(2.7,1.8)
\pcline[linewidth=1pt](1.5,.8)(2.9,1.2)
%6th pending edge
\pcline[linewidth=1pt](-0.1,1.8)(.9,.8)
\pcline[linewidth=1pt](.4,2.4)(1.4,1.4)
%joining lines
\pcline[linecolor=blue, linewidth=1pt](-2.2,-0.4)(-1.4,-0.8)
\pcline[linecolor=blue, linewidth=1pt](-1.9,0.3)(-.9,-0.2)
\pcline[linecolor=blue, linewidth=1pt](-.6,-0.9)(1.4,-.2)
\pcline[linecolor=blue, linewidth=1pt](1.05,.4)(-.9,-0.2)
\pcline[linecolor=blue, linewidth=1pt](1.5,.8)(1.7,.3)
\pcline[linecolor=blue, linewidth=1pt](.9,.8)(1.05,.4)
%joining curve1
\psbezier[linecolor=white, linewidth=14pt](.25,2.1)(-2.15,4.4)(-2.7,4.7)(-4.2,3.2)
\psbezier[linecolor=white, linewidth=14pt](-4.2,3.2)(-5.8,1.6)(-5.4,-3.05)(-3,-0.65)
\psbezier[linestyle=dashed, linewidth=1pt](-.1,1.8)(-2.1,3.8)(-2.7,4.3)(-4,3)
\psbezier[linestyle=dashed, linewidth=1pt](-4,3)(-5.3,1.7)(-5.2,-2.4)(-3.2,-0.4)
\psbezier[linestyle=dashed, linewidth=1pt](.4,2.4)(-2.2,5)(-2.7,5.1)(-4.4,3.4)
\psbezier[linestyle=dashed, linewidth=1pt](-4.4,3.4)(-6.3,1.5)(-5.6,-3.7)(-2.8,-0.9)
%joining curve2
\psbezier[linecolor=white, linewidth=14pt](2.8,1.5)(5.7,2.4)(2.5,4.75)(0,4.75)
\psbezier[linecolor=white, linewidth=14pt](0,4.75)(-2.5,4.75)(-5,2.9)(-2.85,.75)
\psbezier[linestyle=dashed, linewidth=1pt](2.7,1.8)(5.3,2.6)(2,4.5)(0,4.5)
\psbezier[linestyle=dashed, linewidth=1pt](0,4.5)(-2,4.5)(-4.6,3)(-2.6,1)
\psbezier[linestyle=dashed, linewidth=1pt](2.9,1.2)(6.15,2.2)(3,5)(0,5)
\psbezier[linestyle=dashed, linewidth=1pt](0,5)(-3,5)(-5.4,2.8)(-3.1,0.5)
%rotation points $s_i$
%\rput(-2.75,0.6){\makebox(0,0){\tiny $\bullet$}}
%\rput(-0.55,1.7){\makebox(0,0){\tiny $\bullet$}}
%\rput(2.35,1.45){\makebox(0,0){\tiny $\bullet$}}
\rput(3,-0.65){\makebox(0,0){$\star$}}
%\pscircle[linestyle=dashed, linecolor=blue, linewidth=0.5pt](3,-.65){0.4}
\rput(-1.15,-2.5){\makebox(0,0){$\star$}}
%\pscircle[linestyle=dashed, linecolor=blue, linewidth=0.5pt](-1.15,-2.5){0.5}
%\rput(-2.8,-0.45){\makebox(0,0){\tiny $\bullet$}}
%geodesic $\tr F_2F_4$
%\psarc[linecolor=red, linewidth=0.5pt](-2.8,-.45){.2}{125}{305}
%\pcline[linecolor=red, linewidth=0.5pt](-2.94,-0.31)(-2.44,0.04)
%\pcline[linecolor=red, linewidth=0.5pt](-2.75,-0.65)(-2.2,-0.27)
%\pcline[linecolor=red, linewidth=0.5pt](-2.2,-0.27)(-1.4,-0.67)
%\pcline[linecolor=red, linewidth=0.5pt](-1.9,0.1)(-.9,-0.4)
%\pcline[linecolor=red, linewidth=0.5pt](-.9,-0.4)(1.05,.25)
%\pcline[linecolor=red, linewidth=0.5pt](-.6,-0.75)(1.4,-.05)
%\pcline[linecolor=red, linewidth=0.5pt](1.4,-.05)(2.91,-0.83)
%\pcline[linecolor=red, linewidth=0.5pt](1.53,.25)(3.08,-0.47)
%\psarc[linecolor=red, linewidth=0.5pt](3,-.65){.2}{-115}{65}
%\psbezier[linecolor=red, linewidth=0.5pt](-2.44,0.04)(-2.14,.19)(-2.1,0.16)(-1.9,0.1)
%\psbezier[linecolor=red, linewidth=0.5pt](-1.4,-0.67)(-1,-.87)(-1,-0.89)(-.6,-0.75)
%\psbezier[linecolor=red, linewidth=0.5pt](1.05,.25)(1.25,.315)(1.29,0.35)(1.53,.25)
\rput(-1,0.1){\makebox(0,0)[cb]{$Y_2$}}
\rput(1,-0.6){\makebox(0,0)[ct]{$Y_3$}}
\rput(2.3,0.4){\makebox(0,0)[cb]{$Y_4$}}
\rput(-4,-2.3){\makebox(0,0)[rt]{$A$}}
\rput(4.7,2.6){\makebox(0,0)[rt]{$B$}}
\rput(2.8,-1.3){\makebox(0,0)[lt]{$Z_2$}}
\rput(-0.4,-2.3){\makebox(0,0)[lt]{$Z_1$}}
\end{pspicture}
}
\caption{\small
The fat graph corresponding to the ideal triangle partition of the fundamental domain in Fig.~\ref{fi:disc}.
The real numbers $Z_1$ and $Z_2$ are associated to the pending edges and $A$, $B$, and
$Y_2$, $Y_3$, and $Y_4$ to the inner edges.}
\label{fi:treegraph}
\end{figure}

\subsection{The Fuchsian group $\Delta_{g,s,r}$ and geodesic functions}\label{ss:geodesic}

We now describe combinatorially the conjugacy classes of the Fuchsian group $\Delta_{g,s,r}$.
Every time the path homeomorphic to a (closed) geodesic $\gamma$ goes through the edge with the label $\alpha$ we
insert~\cite{Fock1} the matrix of the M\"obius transformation
\be
\label{XZ} X_{Z_\alpha}=\left(
\begin{array}{cc} 0 & -\e^{Z_\alpha/2}\\
                \e^{-Z_\alpha/2} & 0\end{array}\right).
\ee

We also have the ``right'' and ``left'' turn matrices
to be set in proper places when a path makes the corresponding turns at three-valent vertices,
\be
\label{R}
R=\left(\begin{array}{cc} 1 & 1\\ -1 & 0\end{array}\right), \qquad
L= R^2=\left(\begin{array}{cc} 0 & 1\\ -1 &
-1\end{array}\right).
\ee

New elements of the
Fuchsian group correspond to rotations of geodesic lines when going around orbifold points
indicated by star-vertices;
for a ${\mathbb Z}_p$ orbifold point we then insert the matrix $F_p$ (\ref{F-p})
into the corresponding string of $2\times2$-matrices (when we go around the orbifold point counterclockwise
as in Fig.~\ref{fi:corner}(a)). When the order of the orbifold point is larger than two, we can go around it
$k$ times; {\em we then have to insert the matrix $(-1)^{k+1}F_p^k$
into the product of $2\times2$-matrices}.
For example, parts of geodesic functions
in the three cases in Fig.~\ref{fi:corner} read
\be
\label{XZFXZ}
\begin{array}{ll}
\hbox{(a)}\quad & \dots  X_XLX_ZF_pX_ZLX_Y\dots, \\
\hbox{(b)}\quad & \dots  X_XLX_Z(-F^2_p)X_ZRX_X\dots, \\
\hbox{(c)}\quad & \dots  X_YRX_Z(F^3_p)X_ZLX_Y\dots. \\
\end{array}
\ee
Note that $F_p^p=(-1)^{p-1}{\mathbb E}$,
so going around the ${\mathbb Z}_p$ orbifold point $p$ times merely corresponds to avoiding this orbifold point
due to the simple equality (note that $X_S^2=-{\mathbb E}$ and $L^2=-R$)
$$
X_XLX_Z(-1)^{p-1}F_p^pX_ZLX_Y=X_XLX_Z^2LX_Y=-X_XL^2X_Y=X_XRX_Y.
$$
(For the ${\mathbb Z}_2$ orbifold points
this pattern was first proposed by Fock and Goncharov \cite{FG}; the graph morphisms were described in \cite{Ch2}.)

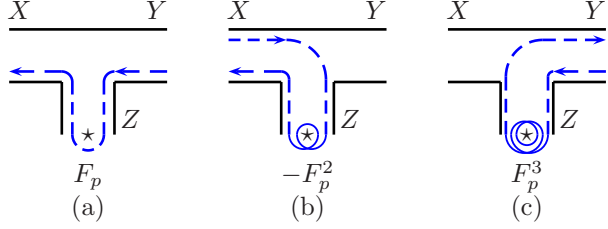
\begin{figure}[tb]
%\hspace*{2cm}
%\epsfysize=6cm
%\vskip .2in
{\psset{unit=0.7}
\begin{pspicture}(-2,-3)(2,2)
\pcline[linewidth=1pt](-1.5,1)(1.5,1)
\pcline[linewidth=1pt](-1.5,0)(-0.5,0)
\pcline[linewidth=1pt](1.5,0)(0.5,0)
\pcline[linewidth=1pt](-.5,0)(-.5,-1)
\pcline[linewidth=1pt](.5,0)(.5,-1)
\rput(0,-1){\makebox(0,0){$\star$}}
\pcline[linecolor=blue, linestyle=dashed, linewidth=1pt]{<-}(-1.5,0.2)(-.5,0.2)
\psarc[linecolor=blue, linestyle=dashed, linewidth=1pt](-.5,0){.2}{0}{90}
\pcline[linecolor=blue, linestyle=dashed, linewidth=1pt](-.3,0)(-.3,-1)
\psarc[linecolor=blue, linestyle=dashed, linewidth=1pt](0,-1){.3}{-180}{0}
\pcline[linecolor=blue, linestyle=dashed, linewidth=1pt](.3,0)(.3,-1)
\psarc[linecolor=blue, linestyle=dashed, linewidth=1pt](.5,0){.2}{90}{180}
\pcline[linecolor=blue, linestyle=dashed, linewidth=1pt]{<-}(.5,0.2)(1.5,0.2)
\rput(-1.3,1.2){\makebox(0,0)[cb]{$X$}}
\rput(1.3,1.2){\makebox(0,0)[cb]{$Y$}}
\rput(0.8,-.7){\makebox(0,0){$Z$}}
\rput(0,-1.8){\makebox(0,0){$F_p$}}
\rput(0,-2.4){\makebox(0,0){(a)}}
\end{pspicture}
\begin{pspicture}(-2,-3)(2,2)
\pcline[linewidth=1pt](-1.5,1)(1.5,1)
\pcline[linewidth=1pt](-1.5,0)(-0.5,0)
\pcline[linewidth=1pt](1.5,0)(0.5,0)
\pcline[linewidth=1pt](-.5,0)(-.5,-1)
\pcline[linewidth=1pt](.5,0)(.5,-1)
\rput(0,-1){\makebox(0,0){$\star$}}
\pcline[linecolor=blue, linestyle=dashed, linewidth=1pt]{<-}(-1.5,0.2)(-.5,0.2)
\psarc[linecolor=blue, linestyle=dashed, linewidth=1pt](-.5,0.05){.15}{0}{90}
\pcline[linecolor=blue, linestyle=dashed, linewidth=1pt](-.35,0.05)(-.35,-1)
\psarc[linecolor=blue, linewidth=0.7pt](0.075,-1){0.275}{-180}{0}
\psarc[linecolor=blue, linewidth=0.7pt](-0.075,-1){0.275}{-180}{0}
\psarc[linecolor=blue, linewidth=0.7pt](0,-1){0.2}{0}{180}
\pcline[linecolor=blue, linestyle=dashed, linewidth=1pt](.35,0.05)(.35,-1)
\psarc[linecolor=blue, linestyle=dashed, linewidth=1pt](-.4,0.05){.75}{0}{90}
\pcline[linecolor=blue, linestyle=dashed, linewidth=1pt]{->}(-1.5,0.8)(-.4,0.8)
\rput(-1.3,1.2){\makebox(0,0)[cb]{$X$}}
\rput(1.3,1.2){\makebox(0,0)[cb]{$Y$}}
\rput(0.8,-.7){\makebox(0,0){$Z$}}
\rput(0,-1.8){\makebox(0,0){$-F^2_p$}}
\rput(0,-2.4){\makebox(0,0){(b)}}
\end{pspicture}
\begin{pspicture}(-2,-3)(2,2)
\pcline[linewidth=1pt](-1.5,1)(1.5,1)
\pcline[linewidth=1pt](-1.5,0)(-0.5,0)
\pcline[linewidth=1pt](1.5,0)(0.5,0)
\pcline[linewidth=1pt](-.5,0)(-.5,-1)
\pcline[linewidth=1pt](.5,0)(.5,-1)
\rput(0,-1){\makebox(0,0){$\star$}}
\pcline[linecolor=blue, linestyle=dashed, linewidth=1pt]{->}(1.5,0.2)(.5,0.2)
\psarc[linecolor=blue, linestyle=dashed, linewidth=1pt](.5,0.1){.1}{90}{180}
\pcline[linecolor=blue, linestyle=dashed, linewidth=1pt](-.4,0.1)(-.4,-1)
\psarc[linecolor=blue, linewidth=0.7pt](0.05,-1){.35}{-180}{0}
\psarc[linecolor=blue, linewidth=0.7pt](-0.05,-1){.35}{-180}{0}
\psarc[linecolor=blue, linewidth=0.7pt](0,-1){.2}{-180}{0}
\psarc[linecolor=blue, linewidth=0.7pt](0.05,-1){.25}{0}{180}
\psarc[linecolor=blue, linewidth=0.7pt](-0.05,-1){.25}{0}{180}
\pcline[linecolor=blue, linestyle=dashed, linewidth=1pt](.4,0.1)(.4,-1)
\psarc[linecolor=blue, linestyle=dashed, linewidth=1pt](.3,0.1){.7}{90}{180}
\pcline[linecolor=blue, linestyle=dashed, linewidth=1pt]{<-}(1.5,0.8)(.3,0.8)
\rput(-1.3,1.2){\makebox(0,0)[cb]{$X$}}
\rput(1.3,1.2){\makebox(0,0)[cb]{$Y$}}
\rput(0.8,-.7){\makebox(0,0){$Z$}}
\rput(0,-1.8){\makebox(0,0){$F^3_p$}}
\rput(0,-2.4){\makebox(0,0){(c)}}
\end{pspicture}
}
\caption{\small Part of a graph with a pending edge.
Its endpoint with the orbifold point is directed toward the interior
of the boundary component this point is associated with.
The variable $Z$ corresponds to the respective pending edge. We present
four typical examples of geodesics undergoing single (a), double (b),
and triple (c) rotations at the ${\mathbb Z}_p$ orbifold point.}
\label{fi:corner}
\end{figure}

An element of a Fuchsian group has then the typical structure
\be
\label{Pgamma}
P_{\gamma}=LX_{Y_n}RX_{Y_{n-1}}\cdots RX_{Y_2}LX_{Z_1}(-1)^{k+1}F^k_p X_{Z_1}RX_{Y_1},
\ee
where $Y_i$ are variables of ``internal'' edges and $Z_j$ are those of pending edges.
The corresponding {\em geodesic function}
\be
\label{G}
G_{\gamma}\equiv \tr P_\gamma=2\cosh(\ell_\gamma/2)
\ee
is expressed via the actual length $\ell_\gamma$ of the closed
geodesic on the Riemann surface.

\begin{remark}\label{rm-positivity}
Note that the combinations
$$
RX_y=\left(\begin{array}{cc} e^{-Y/2} & -e^{Y/2}\\ 0 & e^{Y/2}\end{array}\right)\quad \hbox{and}\quad
LX_y=\left(\begin{array}{cc} e^{-Y/2} & 0 \\ -e^{-Y/2} & e^{Y/2}\end{array}\right)
$$
as well as products of any number of these matrices have the sign structure
$\left(\begin{array}{cc} + & -\\ - & +\end{array}\right)$, so the trace of any of $P_\gamma$ in the absence of
orbifold points is a
sum of exponentials with positive integer coefficients; this sum always include the terms
$e^{Y_1/2+\cdots+Y_n/2}$ and $e^{-Y_1/2-\cdots-Y_n/2}$ being therefore always greater or equal
two thus describing a hyperbolic or parabolic element; the latter is possible only if
$Y_1+\cdots+Y_n=0$ and only if the turn matrices in (\ref{Pgamma}) are all either $R$ or $L$,
which corresponds to a path going along the boundary of a face; all such paths are homeomorphic
to the hole boundaries, and the condition that the sum of $Y_i$ equals zero indicates the degeneration
of a hole into a puncture.
\end{remark}

The group generated by elements (\ref{F-p}) together with translations along $A$- and $B$-cycles
and around holes not necessarily produces a regular (metrizable) surface because its action
is not necessarily discrete. We formulate the necessary and
sufficient conditions for producing a {\em regular} surface
in terms of graphs (see \cite{Ch1a} for the ${\mathbb Z}_2$ orbifold point
case).\footnote{In what follows, we call a Riemann surface
regular if it is locally a smooth
constant-curvature surface everywhere except exactly $r$ orbifold points.}

To formulate the regularity condition, we
interpret passages around orbifold points as paths in the $p$-fold covering of the geodesic neighborhood in Fig.~\ref{fi:saucer-pan}.
For this, we take the subgraph in Fig.~\ref{fi:p-gone-dual} dual to the ideal triangle decomposition in
Fig.~\ref{fi:p-gone}. When splitting the equilateral $p$-gone into ideal triangles
we break the $p$-fold symmetry, so now the shear coordinates
$Z_i$, $i=1,\dots,p$, on the fat graph edges dual to the corresponding
$p$-gone sides and $Y_j$, $j=2,\dots,p-2$, on the edges dual to the diagonals of the $p$-gone are different.

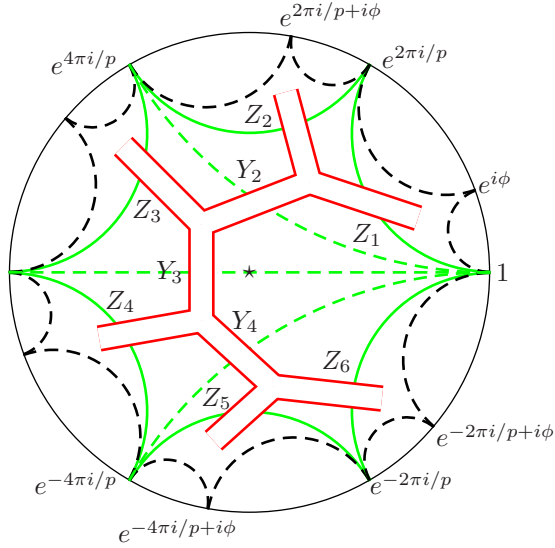
\begin{figure}[tb]
{\psset{unit=0.8}
\begin{pspicture}(5,-5)(-5,5)
%rotated closed geodesics
\pscircle[linewidth=0.5pt](0,0){4}
%curves, bounding the fundamental region
\newcommand{\PATTERN}{%
\psarc[linecolor=green, linewidth=1pt](4.62,0){2.3}{120}{240}
%\rput{10}(0,0){\rput(2.3,0){\makebox(0,0)[rc]{$c$}}}
\rput{-20}(0,0){\psarc[linestyle=dashed, linewidth=1pt](4.06,0){.7}{100}{260}}
\rput{10}(0,0){\psarc[linestyle=dashed, linewidth=1pt](4.26,0){1.45}{110}{250}}
%\rput{-20}(0,0){\rput(3.9,0){\makebox(0,0)[rc]{$a$}}}
%\rput{10}(0,0){\rput(3.8,0){\makebox(0,0)[rc]{$b$}}}
}
\rput{30}(0,0){\PATTERN}
\rput{90}(0,0){\PATTERN}
\rput{150}(0,0){\PATTERN}
\rput{210}(0,0){\PATTERN}
\rput{270}(0,0){\PATTERN}
\rput{330}(0,0){\PATTERN}
\rput{60}(0,0){
\psarc[linestyle=dashed,linecolor=green, linewidth=1pt](8,0){6.93}{150}{210}
}
\rput{-60}(0,0){
\psarc[linestyle=dashed,linecolor=green, linewidth=1pt](8,0){6.93}{150}{210}
}
\pcline[linestyle=dashed,linecolor=green, linewidth=1pt](-4,0)(4,0)
\rput(0,0){\makebox(0,0){$\star$}}
\pcline[linecolor=red, linewidth=10pt](-0.8,0.8)(-0.8,-0.8)
\pcline[linecolor=red, linewidth=10pt](-0.8,-0.8)(-2.5,-1.1)
\pcline[linecolor=red, linewidth=10pt](-0.8,0.8)(-2.1,2.1)
\pcline[linecolor=red, linewidth=10pt](-0.8,0.8)(1,1.5)
\pcline[linecolor=red, linewidth=10pt](-0.8,-0.8)(0.4,-1.9)
\pcline[linecolor=red, linewidth=10pt](1,1.5)(0.6,3)
\pcline[linecolor=red, linewidth=10pt](1,1.5)(2.8,0.9)
\pcline[linecolor=red, linewidth=10pt](0.4,-1.9)(-0.6,-2.8)
\pcline[linecolor=red, linewidth=10pt](0.4,-1.9)(2.2,-2.1)
% second set of line, to make double lines
\pcline[linecolor=white, linewidth=8pt](-0.8,0.8)(-0.8,-0.8)
\pcline[linecolor=white, linewidth=8pt](-0.8,-0.8)(-2.5,-1.1)
\pcline[linecolor=white, linewidth=8pt](-0.8,0.8)(-2.1,2.1)
\pcline[linecolor=white, linewidth=8pt](-0.8,0.8)(1,1.5)
\pcline[linecolor=white, linewidth=8pt](-0.8,-0.8)(0.4,-1.9)
\pcline[linecolor=white, linewidth=8pt](1,1.5)(0.6,3)
\pcline[linecolor=white, linewidth=8pt](1,1.5)(2.8,0.9)
\pcline[linecolor=white, linewidth=8pt](0.4,-1.9)(-0.6,-2.8)
\pcline[linecolor=white, linewidth=8pt](0.4,-1.9)(2.2,-2.1)
\rput(4.1,0){\makebox(0,0)[lc]{$1$}}
\rput(3.8,1.5){\makebox(0,0)[lc]{$e^{i\phi}$}}
\rput(2.2,3.4){\makebox(0,0)[lb]{$e^{2\pi i/p}$}}
\rput(0.5,4.1){\makebox(0,0)[lb]{$e^{2\pi i/p+i\phi}$}}
\rput(-2.2,3.3){\makebox(0,0)[rb]{$e^{4\pi i/p}$}}
\rput(3.1,-2.5){\makebox(0,0)[lt]{$e^{-2\pi i/p+i\phi}$}}
\rput(2,-3.3){\makebox(0,0)[lt]{$e^{-2\pi i/p}$}}
\rput(-0.2,-4.1){\makebox(0,0)[rt]{$e^{-4\pi i/p+i\phi}$}}
\rput(-2.2,-3.3){\makebox(0,0)[rt]{$e^{-4\pi i/p}$}}
\rput(0,1.5){\makebox(0,0)[cb]{$Y_2$}}
\rput(-1.1,0){\makebox(0,0)[rc]{$Y_3$}}
\rput(-0.3,-1){\makebox(0,0)[lb]{$Y_4$}}
\rput(2.2,.8){\makebox(0,0)[rt]{$Z_1$}}
\rput(-0.1,2.6){\makebox(0,0)[lc]{$Z_2$}}
\rput(-1.7,1.2){\makebox(0,0)[ct]{$Z_3$}}
\rput(-1.9,-0.3){\makebox(0,0)[rt]{$Z_4$}}
\rput(-0.3,-2.1){\makebox(0,0)[rc]{$Z_5$}}
\rput(1.2,-1.3){\makebox(0,0)[lt]{$Z_6$}}
\end{pspicture}
}
\caption{\small
The (tree-like) subgraph dual to the ideal triangle decomposition of the ideal equilateral $p$-gone in
Fig.~\ref{fi:p-gone}. All the variables $Z_\alpha$ and $Y_\beta$ are determined by the cross-ratio relations
in the corresponding ideal quadrangles based on the points from the sets $\{e^{2\pi i k/p}\}$ and $\{e^{2\pi i k/p+i\phi}\}$,
$k=0,\dots,p-1$.}
\label{fi:p-gone-dual}
\end{figure}

We first identify the parameter $Z$ in (\ref{XZFXZ}) to be
\be
\label{eZ}
e^Z=\frac{\sin\left(\pi/p-\phi/2\right)}{\sin(\phi/2)},
\ee
We can then derive the explicit relations between the parameter $Z$ in (\ref{eZ}) and the
variables $Z_i$ and $Y_j$ determined by the standard cross-ratio relations (\ref{four-term}).
The vertices of the ideal $p$-gone
are situated at the points $e^{i2\pi k/p}$, $k=0,\dots,p-1$, and $p$ copies of the vertex
of an additional ideal triangle adjacent to the $p$-gone side are $e^{i\phi+i2\pi ik/p}$, $k=0,\dots,p-1$.

Using the cross-ratio relations to calculate $Z_i$ and $Y_j$ we find the exact relations between these variables and
the variable $Z$ given by (\ref{eZ}):
\bea
e^{Z_1}&=&e^Z\frac{\sin(2\pi/p)}{\sin(\pi/p)},\nonumber\\
e^{Z_p}&=&e^Z\frac{\sin(\pi/p)}{\sin(2\pi/p)},\nonumber\\
e^{Z_k}&=&e^Z\frac{\sin\bigl((k-1)\pi/p\bigr)}{\sin(k\pi/p)},\quad k=2,\dots,p-1,\label{ZK}\\
e^{Y_k}&=&\frac{\sin\bigl((k+1)\pi/p\bigr)}{\sin\bigl((k-1)\pi/p\bigr)}, \quad k=2,\dots,p-2.\label{YK}
\eea
The following $2\times 2$-matrix equalities can be verified directly:
\be
\label{equivalence}
X_ZF_pX_Z=X_{Z_1}LX_{Z_2}=X_{Z_{k}}LX_{Y_{k}}LX_{Z_{k+1}}=X_{Z_{p-1}}LX_{Z_p}, \ k=2,\dots,p-2.
\ee
We then have the following lemma.

\begin{lm}\label{lem-length}
We have the following explicit $2\times 2$-matrix relations for the shear variables in the
equilateral $p$-gone in Fig.~\ref{fi:p-gone-dual} given by (\ref{ZK}) and (\ref{YK}):
\bea
X_ZF_pX_Z&=&X_{Z_1}LX_{Z_2}\nonumber\\
X_Z(-F_p^2)X_Z&=&X_{Z_1}RX_{Y_2}LX_{Z_3}\nonumber\\
%X_Z(F_p^3)X_Z&=&X_{Z_4}LX_{Y_3}RX_{Y_2}RX_{Z_1}\\
\vdots&{}&\label{P-gone->pending}\\
X_Z(-1)^{k-1}F_p^{k}X_Z&=&X_{Z_1}RX_{Y_2}R\cdots RX_{Y_k}LX_{Z_{k+1}}, \ k=2,\dots,p-2\nonumber\\
X_Z(-1)^{p}F_p^{p-1}X_Z&=&X_{Z_1}R X_{Y_{2}}R\cdots RX_{Y_{p-2}}R X_{Z_p}. \nonumber
\eea
\end{lm}

The {\it proof} uses equalities from (\ref{equivalence}) for constructing the longer chain using that
$X_SX_S=-{\mathbb E}$ for any variable $S$ and that $L^2=-R$. For example, we obtain the r.h.s. of the second equality
in (\ref{P-gone->pending}) multiplying $X_{Z_1}LX_{Z_2}\cdot X_{Z_{2}}LX_{Y_{2}}LX_{Z_{3}}$ whereas the l.h.s.
merely becomes $X_ZF_pX_Z\cdot X_ZF_pX_Z=X_Z(-1)F_p^2 X_Z$. All other equalities are obtained if we continue this chain
of multiplications.

Due to Lemma~\ref{lem-length}, all ``rotations'' about orbifold points $X_Z(-1)^{k+1}F_p^k X_Z$ are now presented as the standard products of
matrices $X_S$ (with real $S$) alternated with the matrices of left and right turns (\ref{R}), which means that the conclusion of
Remark~\ref{rm-positivity} remains valid in this case as well: as soon as in the original spine $\Gamma_{g,s,r}$
all the parameters $Z_\alpha$ of pending edges are real, all the geodesic functions constructed on a Riemann surface with orbifold points are
Laurent polynomials with {\em positive integer coefficients}
of the ``new'' real variables $Z^{(\alpha,p)}_i$, $Y^{(p)}_j$ (where the superscripts $\alpha,p$ indicate that these variables are completely
determined by the original variable $Z_\alpha$ and the order $p$ of the orbifold point)
and ``old'' variables of ``internal'' edges of the spine $\Gamma_{g,s,r}$.
corresponding to usual partitions into ideal triangles.
So, again, in the trace of every product of form (\ref{Pgamma})
we necessarily have the term $2\cosh (\sum_{\beta=1}^n
X_\beta)$, where the sum ranges all edges (new and internal ones) the corresponding path goes through, and we let $X$ denote
the variables of all these edges disregarding their origins.
Every such trace is therefore a positive number greater or
equal two, and the corresponding element of the group will be either
hyperbolic or parabolic (the latter is possible only for geodesics around holes and only if a hole
reduces to a puncture). The only elliptic elements are precisely conjugates of $F_p^k$. We therefore come to the theorem

\begin{theorem}
\label{lem-metric}
We have a metrizable Riemann surface for {\em any} choice of real numbers
$Z_\alpha$ associated to the edges of an original spine $\Gamma_{g,s,r}$.
The converse statement is also true: for
{\em any} metrizable Riemann surface $\Sigma_{g,s,r}$ we have a spine $\Gamma_{g,s,r}$
with real numbers associated to its edges such that the lengths of geodesics on $\Sigma_{g,s,r}$
are given by traces of products (\ref{Pgamma}) corresponding to paths in the spine.
\end{theorem}

The {\em proof} of the second statement was performed in \cite{Ch2} for ${\mathbb Z}_2$-orbifold points. It
is based on the (obvious) existence of the ideal triangle decomposition described in Sec.~\ref{ss:clusters} for any metrizable
Riemann surface and can be straightforwardly generalized to the case of
orbifold points of any type. We have therefore parameterized all possible regular surfaces in terms of the
$(6g-6+3s+2r)$-tuple of real coordinates $\{Z_\alpha\}$.

\begin{corollary}
The decorated Teichm\"uller space ${\mathfrak T}^{H}_{g,s,r}$ of Riemann surfaces
with holes and orbifold points
is the space $\RR^{6g-6+3s+2r}$ of real parameters
on the edges of a spine $\Gamma_{g,s,r}$.
\end{corollary}

\section{Mapping class group transformations}\label{s:mcg}

\subsection{Poisson structure}\label{ss:Poisson}

One of the most attractive properties of the graph description is a very simple Poisson algebra on the set
of parameters $Z_\alpha$. The following result is the straightforward generalization of the theorem
formulated for surfaces without marked points in~\cite{Fock1} and for surfaces with order-2 orbifold
points in~\cite{FG} (see also \cite{Ch1}).

\begin{theorem}\label{th-WP} In the coordinates $Z_\alpha $ on any fixed spine
corresponding to a surface with orbifold points,
the Weil--Petersson bracket $B_{{\mbox{\tiny WP}}}$ is given by
\be
\label{WP-PB}
\bigl\{f({\mathbf Z}),g({\mathbf Z})\bigr\}=\sum_{{\hbox{\small 3-valent} \atop \hbox{\small vertices $\alpha=1$} }}^{4g+2s+|\delta|-4}
\,\sum_{i=1}^{3\!\!\mod 3}
\left(\frac{\partial f}{\partial Z_{\alpha_i}} \frac{\partial g}{\partial Z_{\alpha_{i+1}}}
- \frac{\partial g}{\partial Z_{\alpha_i}} \frac{\partial f}{\partial Z_{\alpha_{i+1}}}\right),
\ee
where the sum ranges all the three-valent vertices of a graph and
$\alpha_i$ are the labels of the cyclically (counterclockwise)
ordered ($\alpha_4\equiv \alpha_1 $) edges incident to the vertex
with the label $\alpha$ irrespectively on whether these edges are internal or
pending edges of the graph. This bracket gives rise to the {\em Goldman
bracket} on the space of geodesic length functions \cite{Gold}.
\end{theorem}

We identify the exchange matrix $B$ with the matrix of the Poisson relations for the variables $Z_\alpha$.

The center of this Poisson algebra is provided by the proposition.

\begin{prop}\label{prop12}
The center of the Poisson algebra {\rm(\ref{WP-PB})} is generated by
elements of the form $\sum Z_\alpha$, where the sum ranges all edges
of $\Gamma_{g,s,r} $ belonging to the same boundary component
taken with multiplicities.
This means, in particular, that each pending edge, irrespectively on the type of orbifold point it corresponds to,
contributes twice to such sums. The dimension of this center is obviously $s$.
\end{prop}

For the proof in the general case see Appendix~B of~\cite{ChP}. Note that for the path homeomorphic to
the hole boundary, for any number of insertions of matrices $F_{p_i}$ with any $p_i$, we have
\bea
&{}&\hbox{tr\,}\bigl[LX_{Y_1}LX_{Y_2}\cdots LX_{Y_k}F_{p_i}X_{Y_k}L\cdots LX_{Y_{n-1}}LX_{Y_n}\bigr]
\nonumber\\
&{}&\quad = 2\cosh \Bigl[\frac{Y_1}{2}+\frac{Y_2}{2}+\dots +Y_k+\dots +\frac{Y_{n-1}}{2}+\frac{Y_n}{2}\Bigr].
\nonumber
\eea

\subsection{Flip morphisms of fat graphs}\label{ss:flip}

In this section, we present the complete list of mapping class group transformations
that enable us to change numbers $|\delta_k|$ of orbifold points associated with the $k$th
hole, change the cyclic ordering inside any of the sets $\delta_k$, flip any inner edge
of the graph and, eventually, change the orientation of the geodesic
spiraling to the hole perimeter (in the case where we have more than one hole).\footnote{These
transformations are dual to mutations of cluster variables from Sec.~\ref{s-algebra}.}
We can therefore
establish a morphism between any two of the graphs belonging to the same class $\Gamma_{g,s,r}$
with the same (unordered) sets of orbifold point orders $\{p_i\}_{i=1}^r$.

\subsubsection{Whitehead moves on inner edges}\label{sss:mcg}

Given a spine $\Gamma$ of $\Sigma$ and
assuming that the
edge $\alpha$ has distinct endpoints, we may produce
another spine $\Gamma _\alpha$ of $\Sigma$ by contracting and expanding edge $\alpha$ of
$\Gamma $, the edge labeled $Z$ in Figure~\ref{fi:flip}.
This transformation is dual to the mutation (\ref{diagonal}).
We say that $\Gamma _\alpha$ arises from $\Gamma$ by a
{\it Whitehead move} (or flip) along the edge $\alpha$.
A labeling of edges of the spine $\Gamma$ implies a natural labeling of edges of the
spine $\Gamma_\alpha$; we then obtain a morphism between the spines $\Gamma$ and $\Gamma_\alpha$.

\begin{figure}[tb]
\setlength{\unitlength}{1.5mm}%
\begin{picture}(-60,27)(60,48)
\thicklines
\put(32,62){\line( -1,2){ 4}}
\put(36,64){\line( -1,2){ 4}}
\put(32,62){\line(-1,-2){ 4}}
\put(36,60){\line(-1,-2){ 4}}
\put(36,60){\line( 1, 0){20}}
\put(36,64){\line( 1, 0){20}}
\put(60,62){\line( 1, 2){ 4}}
\put(60,62){\line( 1,-2){ 4}}
\put(56,64){\line( 1, 2){ 4}}
\put(56,60){\line( 1,-2){ 4}}
\thinlines
\put(21,62){\vector(-1, 0){  0}}
\put(21,62){\vector( 1, 0){ 5}}
\thicklines
\put(10,54){\line( -2,-1){ 8}}
\put(8,58){\line( -2,-1){ 8}}
\put(8,58){\line( 0,1){8}}
\put(12,58){\line( 0,1){8}}
\put(10,54){\line(2,-1){ 8}}
\put(12,58){\line(2,-1){ 8}}
\put(10,70){\line( -2,1){ 8}}
\put(8,66){\line( -2,1){ 8}}
\put(10,70){\line(2,1){ 8}}
\put(12,66){\line(2,1){ 8}}
\put( 4,74){\makebox(0,0)[lb]{$A$}}
\put(16,74){\makebox(0,0)[rb]{$B$}}
\put(14,62){\makebox(0,0)[lc]{$Z$}}
\put(16,50){\makebox(0,0)[rt]{$C$}}
\put( 4,50){\makebox(0,0)[lt]{$D$}}
\put(32,52){\makebox(0,0)[lt]{$D - \phi(-Z)$}}
\put(60,52){\makebox(0,0)[rt]{$C+\phi(Z)$}}
\put(60,73){\makebox(0,0)[rb]{$B-\phi(-Z)$}}
\put(32,73){\makebox(0,0)[lb]{$A+\phi(Z)$}}
\put(47,66){\makebox(0,0)[cb]{$-Z$}}
% added geodesic lines
\color[rgb]{1,0,0}
\thinlines
%left half
\put(0.5,53){\line( 2,1){ 6}}
\put(1,52){\line( 2,1){ 6}}
\put(1.5,51){\line( 2,1){ 6}}
\put(0.5,71){\line( 2,-1){ 6}}
\put(19,72){\line( -2,-1){ 6}}
\put(18.5,51){\line( -2,1){ 6}}
\qbezier(6.5,56)(9,57.75)(9,60)
\qbezier(6.5,68)(9,66.25)(9,64)
\put(9,60){\line( 0,1){ 4}}
\qbezier(7,55)(10,56.5)(10,62)
\qbezier(13,69)(10,67.5)(10,62)
\qbezier(7.5,54)(10,55.75)(12.5,54)
\put(-1.5,72){\makebox(0,0)[cc]{\hbox{\small$1$}}}
\put(-1.5,72){\circle{3}}
\put(21,73){\makebox(0,0)[cc]{\hbox{\small$2$}}}
\put(21,73){\circle{3}}
\put(20.5,50){\makebox(0,0)[cc]{\hbox{\small$3$}}}
\put(20.5,50){\circle{3}}
%right half
\put(29,53.5){\line( 1,2){ 3}}
\put(30,53){\line( 1,2){ 3}}
\put(31,52.5){\line( 1,2){ 3}}
\put(29,70.5){\line( 1,-2){ 3}}
\put(62,71){\line( -1,-2){ 3}}
\put(61,52.5){\line( -1,2){ 3}}
\qbezier(34,58.5)(35.25,61)(38,61)
\qbezier(58,58.5)(56.75,61)(54,61)
\put(38,61){\line( 1,0){ 16}}
\qbezier(33,59)(34.5,62)(38,62)
\qbezier(59,65)(57.5,62)(54,62)
\put(38,62){\line( 1,0){ 16}}
\qbezier(32,59.5)(33.25,62)(32,64.5)
\put(28,72.5){\makebox(0,0)[cc]{\hbox{\small$1$}}}
\put(28,72.5){\circle{3}}
\put(63,73){\makebox(0,0)[cc]{\hbox{\small$2$}}}
\put(63,73){\circle{3}}
\put(62,50.5){\makebox(0,0)[cc]{\hbox{\small$3$}}}
\put(62,50.5){\circle{3}}
\end{picture}
\caption{\small Flip, or Whitehead move on the shear coordinates $Z_\alpha$. The outer edges can be
pending, but the edge with respect to which the morphism is performed must be an internal
edge. We also indicate the correspondences between geodesic paths under the flip.}
\label{fi:flip}
\end{figure}
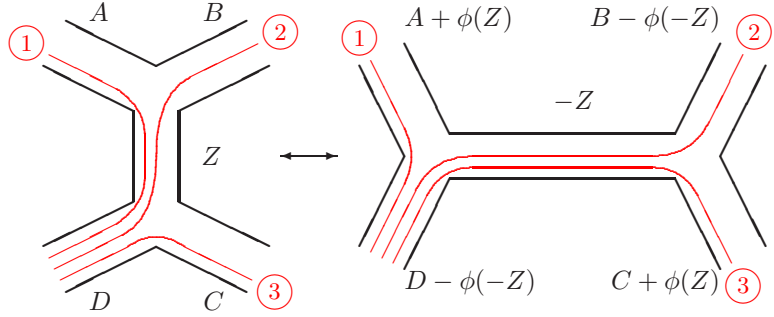

\begin{prop} {\rm \cite{ChF}}\label{propcase}
Setting $\phi (Z)={\rm log}(1+e^Z)$ and adopting the notation of Fig.~\ref{fi:flip}
for shear coordinates of nearby edges, the effect of a
Whitehead move is
\be
W_Z\,:\ (A,B,C,D,Z)\to (A+\phi(Z), B-\phi(-Z), C+\phi(Z), D-\phi(-Z), -Z)
\label{abc}
\ee
In the various cases where the edges are not distinct
and identifying an edge with its shear coordinate in the obvious notation we have:
if $A=C$, then $A'=A+2\phi(Z)$;
if $B=D$, then $B'=B-2\phi(-Z)$;
if $A=B$ (or $C=D$), then $A'=A+Z$ (or $C'=C+Z$);
if $A=D$ (or $B=C$), then $A'=A+Z$ (or $B'=B+Z$).
Any subset of edges $A$, $B$, $C$, and $D$ can be pending edges of the graph.
\end{prop}

We have the lemma establishing the properties of
invariance w.r.t. the flip morphisms~\cite{ChF}.

\begin{lm} \label{lem-abc}
Transformation~{\rm(\ref{abc})} preserves
the traces of products over paths {\rm(\ref{G})} (the geodesic functions) and
transformation~{\rm(\ref{abc})} simultaneously preserves
Poisson structure {\rm(\ref{WP-PB})} on the shear coordinates.
\end{lm}

\subsubsection{Whitehead moves on pending edges}\label{sss:pending}

Choosing other representatives of the orbifold points in the Poincar\'e disc, we obtain different
fundamental domains with different cyclic ordering of the (preimages) of the orbifold points
$s_i$ $(i=1,\dots, |\delta_k|)$ possibly with transferring orbifold points from one set $\delta_k$
to another set $\delta_{k'}$.

Analogously to the mutation in Fig.~\ref{fi:p-gone}, flipping the pending edge corresponds to choosing
another fundamental domain, as shown in Fig.~\ref{fi:interchange-p-dual}. We take there $e^Z$ given by formula (\ref{eZ})
and $e^{Y_{1,2}}$ and $e^{\widetilde{Y_{1,2}}}$ given by the standard cross-ratio relations, for example,
$$
e^{Y_2}=\frac{\left(1-e^{i\phi}\right)\left(e^{ib}-e^{2\pi i/p}\right)}
{\left(e^{i\phi}-e^{ib}\right)\left(1-e^{2\pi i/p}\right)},\qquad
e^{\widetilde{Y_2}}=\frac{\left(e^{-2\pi i/p+ib}-1\right)\left(e^{-2\pi i/p+i\phi}-e^{i\phi}\right)}
{\left(e^{-2\pi i/p+i\phi}-e^{-2\pi i/p+ib}\right)\left(1-e^{i\phi}\right)}.
$$

\begin{figure}[tb]
%\hspace*{2cm}
%\epsfysize=6cm
%\vskip .2in
{\psset{unit=0.8}
\begin{pspicture}(-4,-5)(4,1)
%rotated closed geodesics
\psarc[linewidth=0.5pt](0,0){4}{210}{320}
%\rput(0,0){\makebox(0,0){$\star$}}
\psarc[linestyle=dashed, linecolor=blue, linewidth=1pt](0,0){1.65}{205}{325}
%curves, bounding the fundamental region
\rput(0,0){\psarc[linecolor=green, linewidth=0.5pt](0,-5.65){4}{45}{135}}
\rput{25}(0,0){\psarc[linecolor=green, linewidth=0.5pt](0,-4.25){1.45}{20}{160}}
\rput{-20}(0,0){\psarc[linecolor=green, linewidth=0.5pt](0,-4.41){1.86}{25}{155}}
\rput{-30}(0,0){\psarc[linecolor=green, linewidth=0.5pt](0,-4.14){1.07}{15}{165}}
\rput{35}(0,0){\psarc[linecolor=green, linewidth=0.5pt](0,-4.06){.7}{10}{170}}
\rput{15}(0,0){\psarc[linecolor=green, linewidth=0.5pt](0,-4.06){.7}{10}{170}}
\rput{-5}(0,0){\psarc[linecolor=green, linewidth=0.5pt](0,-4.06){.7}{10}{170}}
%\rput(-1.1,-2){\makebox(0,0){\tiny $\bullet$}}
%
\rput(0,0){
\psline[linewidth=0.5pt](0.2,-1.2)(0.4,-2.2)
\psline[linewidth=0.5pt](-0.3,-1.25)(-0.1,-2.25)
\psline[linewidth=0.5pt](1.6,-2.8)(0.4,-2.2)
\psline[linewidth=0.5pt](-0.8,-2.95)(-0.1,-2.25)
\psline[linewidth=0.5pt](0.2,-2.5)(1.2,-3)
\psline[linewidth=0.5pt](0.2,-2.5)(-0.35,-3.05)
}
\rput(-0.9,-2.5){\makebox(0,0)[cc]{$Y_1$}}
\rput(1.1,-2.1){\makebox(0,0)[cc]{$Y_2$}}
\rput(0,-.9){\makebox(0,0)[cc]{$Z$}}
\rput(1.6,-3.9){\makebox(0,0)[lt]{$e^{ib}$}}
\rput(3,-3){\makebox(0,0)[lt]{$e^{2\pi i/p}$}}
\rput(-3,-3){\makebox(0,0)[rt]{$1$}}
%\rput(3,3){\makebox(0,0)[lb]{$-1$}}
%\rput(-3,3){\makebox(0,0)[rb]{$\frac{-1}{1-w^{-2}}$}}
\rput(0.2,-4.2){\makebox(0,0)[ct]{$e^{i\phi}$}}
\rput(-1.2,-4.1){\makebox(0,0)[ct]{$e^{ia}$}}
%left-side picture: half-circles in upper-half plane
\end{pspicture}
}
{\psset{unit=0.8}
\begin{pspicture}(2,-5)(-5,1)
%rotated closed geodesics
\rput(-6.5,-1){\psline[linewidth=2pt]{->}(0,0)(0.5,0)}
\psarc[linewidth=0.5pt](0,0){4}{165}{280}
%\rput(0,0){\makebox(0,0){$\star$}}
\psarc[linestyle=dashed, linecolor=blue, linewidth=1pt](0,0){1.65}{160}{285}
%curves, bounding the fundamental region
\rput{-40}(0,0){\psarc[linecolor=red, linewidth=0.5pt](0,-5.65){4}{45}{135}}
%\rput{130}(0,0){\psarc[linestyle=dashed, linecolor=red, linewidth=0.5pt](0,-5.65){4}{45}{135}}
%\rput{220}(0,0){\psarc[linestyle=dashed, linecolor=red, linewidth=0.5pt](0,-5.65){4}{45}{135}}
%\rput{285}(0,0){\rput(0,-3){\makebox(0,0){\tcr{$\bullet$}}}}
%\rput{310}(0,0){\rput(0,-3){\makebox(0,0){\tcr{$\bullet$}}}}
%\rput{335}(0,0){\rput(0,-3){\makebox(0,0){\tcr{$\bullet$}}}}
%
\rput{-65}(0,0){\psarc[linecolor=red, linewidth=0.5pt](0,-4.25){1.45}{20}{160}}
\rput{-20}(0,0){\psarc[linecolor=green, linewidth=0.5pt](0,-4.41){1.86}{25}{155}}
\rput{-30}(0,0){\psarc[linecolor=green, linewidth=0.5pt](0,-4.14){1.07}{15}{165}}
\rput{-55}(0,0){\psarc[linecolor=red, linewidth=0.5pt](0,-4.06){.7}{10}{170}}
\rput{-75}(0,0){\psarc[linecolor=red, linewidth=0.5pt](0,-4.06){.7}{10}{170}}
\rput{-5}(0,0){\psarc[linecolor=green, linewidth=0.5pt](0,-4.06){.7}{10}{170}}
\rput(-3,-3){\makebox(0,0)[rt]{$1$}}
\rput(0,-4.2){\makebox(0,0)[lt]{$e^{i\phi}$}}
\rput(-1.2,-4.1){\makebox(0,0)[ct]{$e^{ia}$}}
\rput(-3.9,-1.6){\makebox(0,0)[rt]{$e^{-\frac{2\pi i}{p}+ib}$}}
\rput(-4.2,0){\makebox(0,0)[rt]{$e^{-\frac{2\pi i}{p}+i\phi}$}}
%\rput(0,4.2){\makebox(0,0)[lb]{$b''$}}
%\rput(-4.2,0){\makebox(0,0)[rt]{$b'''$}}
%
\rput{-35}(0,0){
\psline[linewidth=0.5pt](-0.2,-1.2)(-0.4,-2.2)
\psline[linewidth=0.5pt](0.3,-1.25)(0.1,-2.25)
\psline[linewidth=0.5pt](-1.6,-2.8)(-0.4,-2.2)
\psline[linewidth=0.5pt](0.6,-2.75)(0.1,-2.25)
\psline[linewidth=0.5pt](-0.2,-2.5)(-1.2,-3)
\psline[linewidth=0.5pt](-0.2,-2.5)(0.2,-2.9)
}
\rput(-0.8,-2.1){\makebox(0,0)[cc]{${\widetilde{Y_1}}$}}
\rput(-2.2,-1){\makebox(0,0)[cc]{${\widetilde{Y_2}}$}}
\rput(-.5,-0.6){\makebox(0,0)[cc]{${\widetilde{Z}}$}}
\end{pspicture}
}
\caption{\small
The transformation of dual variables ($h$-lengths)
$\{Y_1, Y_2, Z\}\to \{{\widetilde{Y_1}},{\widetilde{Y_2}},{\widetilde{Z}}\}$ described
by (\ref{morphism-pending}) with $w=2\cos(\pi/p)$.
}
\label{fi:interchange-p-dual}
\end{figure}
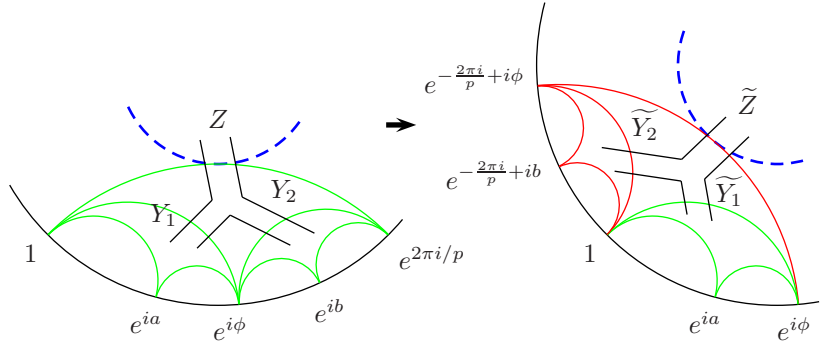

\begin{lm} \label{lem-pending1}
The transformation~in Fig.~\ref{fi:interchange-p-dual} with $e^Z$ given by (\ref{eZ}) has the form
\be
\label{morphism-pending}
\{\tilde Y_1,\tilde Y_2,\tilde Z\}= \{Y_1-\log(1+we^{-Z}+e^{-2Z}),Y_2+\log(1+we^Z+e^{2Z}),-Z\}
\ee
and is the morphism of the space
${\mathcal T}_{g,s,r}^H$. These
morphisms preserve both Poisson structures {\rm(\ref{WP-PB})} and the geodesic
functions. In Fig.~\ref{fi:interchange-p-dual} any (or both) of $Y$-variables can be
variables of pending edges (the transformation formula is insensitive to it).
\end{lm}

{\bf Proof.} Verifying the preservation of Poisson relations (\ref{WP-PB}) is simple, whereas
for traces over paths we have four cases, and in each of these cases we have the following
$2\times2$-{\em matrix} equalities to be verified directly:
\bea
X_{Y_2}LX_ZF_p^kX_ZLX_{Y_1}&=&
-X_{{\tilde Y}_2}RX_{\tilde Z}F_p^{k-1}X_{\tilde Z}RX_{{\tilde Y}_1},\nonumber\\
X_{Y_1}RX_ZF_p^kX_ZLX_{Y_1}&=&
-X_{{\tilde Y}_1}LX_{\tilde Z}F_p^{k}X_{\tilde Z}RX_{{\tilde Y}_1},\nonumber\\
X_{Y_2}LX_ZF_p^kX_ZRX_{Y_2}&=&
-X_{{\tilde Y}_2}RX_{\tilde Z}F_p^{k}X_{\tilde Z}LX_{{\tilde Y}_2}.\nonumber
\eea
%\bea
%X_{Y_2}LX_ZF_p^k(-1)^{k-1}X_ZLX_{Y_1}&=&
%X_{{\tilde Y}_2}RX_{\tilde Z}F_p^{p-k+1}(-1)^{p-k}X_{\tilde Z}RX_{{\tilde Y}_1},\nonumber\\
%X_{Y_1}RX_ZF_p^k(-1)^{k-1}X_ZLX_{Y_1}&=&
%X_{{\tilde Y}_1}LX_{\tilde Z}F_p^{p-k+1}(-1)^{p-k}X_{\tilde Z}RX_{{\tilde Y}_1},\nonumber\\
%X_{Y_2}LX_ZF_p^k(-1)^{k-1}X_ZRX_{Y_2}&=&
%X_{{\tilde Y}_2}RX_{\tilde Z}F_p^{p-k+1}(-1)^{p-k}X_{\tilde Z}LX_{{\tilde Y}_2}.\nonumber
%\eea

Using flip morphisms in Fig.~\ref{fi:interchange-p-dual} and in formula
(\ref{abc}), we establish a morphism between any two algebras
corresponding to surfaces of the same genus, same number of boundary
components, and same numbers of ${\mathbb Z}_p$-orbifold
points of each sort $p$; the distribution of latter into the boundary components as well as
the cyclic ordering inside each of the boundary component can be arbitrary.

It is a standard tool that if, after a
series of morphisms, we come to a graph of the same combinatorial
type as the initial one (disregarding labeling of edges but distinguishing between
different orbifold types of pending vertices), we
associate a {\em mapping class group} operation to this morphism
therefore passing from the groupoid of morphisms to the group of
modular transformations.

\begin{remark}\label{rem-hole}
Another way of interpreting transformations (\ref{morphism-pending}) is as follows.
We can imitate the above
flips/mutations by introducing a new ``hole'' with possibly imaginary perimeter $P$
and considering the following chain of {\em standard} flips:
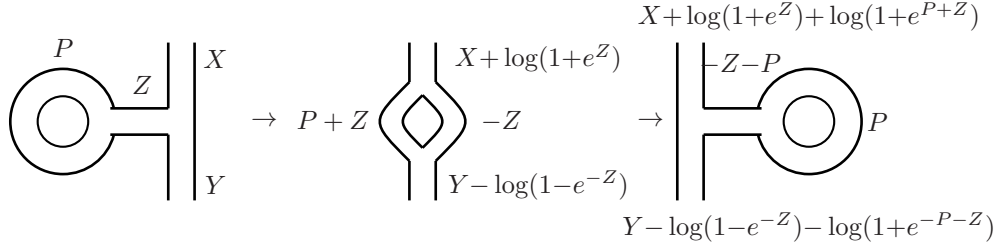
\begin{center}
{\psset{unit=0.7}
\begin{pspicture}(-2.5,-3)(2.5,3)
\pscircle(-1,0){0.5}
\psarc[linewidth=1pt](-1,0){1}{15}{345}
\pcline[linewidth=1pt](-0.1,0.25)(1,0.25)
\pcline[linewidth=1pt](-0.1,-0.25)(1,-0.25)
\pcline[linewidth=1pt](1,1.5)(1,0.25)
\pcline[linewidth=1pt](1,-1.5)(1,-0.25)
\pcline[linewidth=1pt](1.5,-1.5)(1.5,1.5)
\rput(-1,1.4){\makebox(0,0){$P$}}
\rput(.5,0.7){\makebox(0,0){$Z$}}
\rput(1.9,1.2){\makebox(0,0){$X$}}
\rput(1.9,-1.2){\makebox(0,0){$Y$}}
\end{pspicture}
}
{\psset{unit=0.7}
\begin{pspicture}(-3,-3)(3,3)
\rput(-3,0){\makebox(0,0){$\to$}}
\psbezier[linewidth=1pt](0,0.5)(-0.5,0)(-0.5,0)(0,-0.5)
\psbezier[linewidth=1pt](0,0.5)(0.5,0)(0.5,0)(0,-0.5)
\psbezier[linewidth=1pt](-0.25,0.75)(-1,0)(-1,0)(-0.25,-0.75)
\psbezier[linewidth=1pt](0.25,0.75)(1,0)(1,0)(0.25,-0.75)
\pcline[linewidth=1pt](-0.25,.75)(-0.25,1.5)
\pcline[linewidth=1pt](0.25,.75)(0.25,1.5)
\pcline[linewidth=1pt](-0.25,-.75)(-0.25,-1.5)
\pcline[linewidth=1pt](0.25,-.75)(0.25,-1.5)
\rput(-1.7,0){\makebox(0,0){$P+Z$}}
\rput(1.5,0){\makebox(0,0){$-Z$}}
\rput(2.2,1.2){\makebox(0,0){$X{+}\log(1{+}e^Z)$}}
\rput(2.2,-1.2){\makebox(0,0){$Y{-}\log(1{-}e^{-Z})$}}
\end{pspicture}
}
{\psset{unit=0.7}
\begin{pspicture}(-3,-3)(3,3)
\rput(-2,0){\makebox(0,0){$\to$}}
\pscircle(1,0){0.5}
\psarc[linewidth=1pt](1,0){1}{-165}{165}
\pcline[linewidth=1pt](0.1,0.25)(-1,0.25)
\pcline[linewidth=1pt](0.1,-0.25)(-1,-0.25)
\pcline[linewidth=1pt](-1,1.5)(-1,0.25)
\pcline[linewidth=1pt](-1,-1.5)(-1,-0.25)
\pcline[linewidth=1pt](-1.5,-1.5)(-1.5,1.5)
\rput(2.3,0){\makebox(0,0){$P$}}
\rput(-.3,1.1){\makebox(0,0){$-Z{-}P$}}
\rput(1,2){\makebox(0,0){$X{+}\log(1{+}e^Z){+}\log(1{+}e^{P{+}Z})$}}
\rput(1,-2){\makebox(0,0){$Y{-}\log(1{-}e^{-Z}){-}\log(1{+}e^{-P{-}Z})$}}
\end{pspicture}
}
\end{center}

In this pattern, it is useful to shift the variable $Z$ and introduce
$$
{\overline Z}=Z+\frac{P}2.
$$
The transformation for the variables $X$ and $Y$ then just becomes (\ref{morphism-pending}),
\bea
\left[\begin{array}{l}
X \\ Y\\ \overline Z
\end{array}\right]&\to&
\left[\begin{array}{l}
X+\log\bigl[(1+e^{\overline Z-P/2})(1+e^{\overline Z+P/2})\bigr] \\
Y-\log\bigl[(1+e^{-\overline Z+P/2})(1+e^{-\overline Z-P/2})\bigr]\\
-\overline Z
\end{array}\right]=\nn\\
&=&
\left[\begin{array}{l}
X+\log\bigl[1+\omega_p e^{\overline Z}+e^{2\overline Z}\bigr] \\
Y-\log\bigl[1+\omega_p e^{-\overline Z}+e^{-2\overline Z}\bigr]\\
-\overline Z
\end{array}\right],\nn
\eea
where $\omega_p=e^{P/2}+e^{-P/2}$.

Since $P$ is not affected by the above sequence of flips, we can merely erase the corresponding
loop and present it as flipping the pending edge, which carries besides the cluster variable
$\overline Z$ also the coefficient $\omega_p$, which is preserved by mutations
and is equal to $2\cos(\pi/p)$ in the geometric case.
\end{remark}

\subsubsection{Changing the spiraling direction}

The last mapping class group transformation changes the sign of the hole perimeter:
\be
\label{loopinvert}
{\psset{unit=0.7}
\begin{pspicture}(-5,-3)(7,1)
%A3
\pcline[linewidth=1pt](-6,-0.5)(-4,-0.5)
\pcline[linewidth=1pt](-6,-1.5)(-4,-1.5)
\psbezier[linewidth=1pt](-4,-0.5)(-3,1)(-1,1)(-1,-1)
\psbezier[linewidth=1pt](-4,-1.5)(-3,-3)(-1,-3)(-1,-1)
\psbezier[linewidth=1pt](-3.2,-1)(-2.4,-0.2)(-2,-0.3)(-2,-1)
\psbezier[linewidth=1pt](-3.2,-1)(-2.4,-1.8)(-2,-1.7)(-2,-1)
\rput(-5.5,0.2){\makebox(0,0){$Y$}}
\rput(-1.5,1){\makebox(0,0){$P$}}
\pcline[linewidth=1pt]{<->}(0,-1)(2,-1)
%A4
\pcline[linewidth=1pt](3,-0.5)(5,-0.5)
\pcline[linewidth=1pt](3,-1.5)(5,-1.5)
\psbezier[linewidth=1pt](5,-0.5)(6,1)(8,1)(8,-1)
\psbezier[linewidth=1pt](5,-1.5)(6,-3)(8,-3)(8,-1)
\psbezier[linewidth=1pt](5.8,-1)(6.6,-0.2)(7,-0.3)(7,-1)
\psbezier[linewidth=1pt](5.8,-1)(6.6,-1.8)(7,-1.7)(7,-1)
\rput(3.5,0.2){\makebox(0,0){$Y+P$}}
\rput(7.5,1){\makebox(0,0){$-P$}}
\rput(10,0){\makebox(0,0){.}}
%geodesic lines
%left
\pcline[linecolor=red, linewidth=0.5pt]{->}(-6,-.7)(-3.9,-.7)
\pcline[linecolor=red, linewidth=0.5pt](-6,-1.3)(-3.9,-1.3)
\psbezier[linecolor=red, linewidth=0.5pt](-3.9,-.7)(-3,.7)(-1.3,.8)(-1.3,-1)
\psbezier[linecolor=red, linewidth=0.5pt](-3.9,-1.3)(-3,-2.7)(-1.3,-2.8)(-1.3,-1)
%right
\pcline[linecolor=red, linewidth=0.5pt]{->}(3,-.7)(5.1,-.7)
\pcline[linecolor=red, linewidth=0.5pt](3,-1.3)(5.1,-1.3)
\psbezier[linecolor=red, linewidth=0.5pt](5.1,-.7)(6,.7)(7.7,.8)(7.7,-1)
\psbezier[linecolor=red, linewidth=0.5pt](5.1,-1.3)(6,-2.7)(7.7,-2.8)(7.7,-1)
\end{pspicture}
}
\ee

\begin{lm}\label{lem-spiral}
Transformation (\ref{loopinvert}) preserves the Poisson brackets and the set of geodesic functions.
\end{lm}

{\bf Proof}. The preservation of the Poisson bracket is obvious because the variable $P$ Poisson commutes with all
other variables, whereas the preservation of geodesic functions follows from two matrix equalities:
\bea
&&X_YLX_PLX_Y=X_{Y+P}LX_{-P}LX_{Y+P},
\nonumber
\\
&&X_YRX_PRX_Y=X_{Y+P}RX_{-P}RX_{Y+P}.
\nonumber
\eea

We can therefore enlarge the mapping class group of ${\mathcal T}^H_{g,s,r}$ by adding
symmetries between sheets of the $2^s$-ramified covering of the ``genuine'' (nondecorated)
Teichm\"uller space ${\mathcal T}_{g,s,r}$.

The geometrical meaning of this transformation is clear: we change the direction of spiraling to the
hole perimeter line for all lines of the ideal triangle decomposition that spiral to a given hole
like in Fig.~\ref{fi:saucer-pan}.

We can summarize as follows.

\begin{theorem}
The whole mapping class group of $\Sigma_{g,s,r}$ is generated by morphisms described by
Lemmas~\ref{lem-abc},~\ref{lem-pending1}, and~\ref{lem-spiral}.
\end{theorem}

\section{Conjectures}

Lemma~\ref{cc-prime} shows that generalized transformation $a^2+2\cos(\pi/p)ab+b^2$ appears as a flip in the presence of an orbifold point of order $p$. The generalized cluster algebra constructed in this way
is a subalgebra of a bigger standard cluster algebra (maybe of infinite rank)
associated with triangulated surface while generalized exchange relation
are sequences of standard mutations.
Note that the positivity of Laurent polynomials for cluster algebras associated with bordered surfaces is known by ~\cite{MSchW}.
This implies the positivity of Laurent polynomials in generalized cluster algebra associated with triangulations of the surface with arbitrary orbifold points.

We formulate the following conjecture.

\begin{conjecture}\label{generalizedPositivity} If $\rho$ is a reciprocal polynomial with positive coefficients then any cluster variable  of a generalized cluster algebra is expressed as a positive Laurent polynomial in the initial cluster.
\end{conjecture}

We checked by direct inspection that the statement holds for finite type rank 2 cluster algebras.

The example above leads to a natural question:

{\bf Question: } \emph{Is any generalized cluster algebra  a subalgebra of some standard cluster algebra?}

\end{document}